\documentclass[acmsmall,authorversion,nonacm]{acmart}
\usepackage{booktabs} 

\usepackage[ruled]{algorithm2e} 

\SetAlFnt{\small}
\SetAlCapFnt{\small}
\SetAlCapNameFnt{\small}
\SetAlCapHSkip{0pt}
\IncMargin{-\parindent}

\setcitestyle{authoryear}

\usepackage[utf8]{inputenc}

\usepackage{natbib}
\usepackage{caption}

\usepackage{booktabs} 
\usepackage[ruled]{algorithm2e} 
\usepackage{algorithmic} 
\usepackage{pifont}
\usepackage{nicefrac}
\usepackage{color}
\usepackage{xcolor,colortbl}
\usepackage{wrapfig}
\usepackage{graphicx}
\usepackage{amsmath,amsthm}
\usepackage{mathtools}
\usepackage{arydshln}
\usepackage{diagbox}
\usepackage[a]{esvect}
 
\usepackage{bm}
\usepackage{rotating}
\usepackage{multicol}
\usepackage{multirow}
\usepackage{titlesec}
\usepackage{latexsym}
\usepackage{enumitem}
\usepackage{hyperref}[backref=page]
\usepackage{float}
\usepackage{ifthen}

\usepackage{nicematrix}

\hypersetup{
     colorlinks   = true,
     linkcolor    = red, 
     urlcolor     = blue, 
	 citecolor    = blue 
} 

\usepackage[framemethod=TikZ]{mdframed}
\newcounter{theo}[section] 

\newtheorem{thm}{Theorem}
\newtheorem{dfn}{Definition}

\newtheorem{lem}{Lemma}
\newtheorem{ex}{Example}

\newtheorem{prop}{Proposition}
\newtheorem{coro}{Corrollary}
\newenvironment{sketch}{{\sc Proof sketch.}\rm }{\hfill $\Box$ }

\newcounter{newct}

\newtheorem{innercustomthm}{Theorem}
\newenvironment{customthm}[1]
  {\renewcommand\theinnercustomthm{#1}\innercustomthm}
  {\endinnercustomthm}

\newcommand{\bv}{\begin{array}}

\newcommand{\appCoro}[3][]{
\vspace{1mm} {\sc Corollary~\ref{#2}}{\ifthenelse{\equal{#1}{}}{}{ ({\bf #1})}. }{\em #3 \vspace{1mm}}}

\newcommand{\mainmessage}[3][]{
\vspace{1mm} {\bf Main Message~{#2}}{\ifthenelse{\equal{#1}{}}{}{ ({\bf #1})}. }{\em #3 \vspace{1mm}}}

\newcommand{\myparagraph}[1]{\vspace{2mm}\noindent {\bf\boldmath #1}}

\newcommand{\cal}{\mathcal }
\newcommand{\ma}{\mathcal A}

\newcommand{\calZ}{\mathcal Z}
\newcommand{\calY}{\mathcal Y}
\newcommand{\calX}{\mathcal X}

\newcommand{\calT}{\mathcal T}

\newcommand{\calG}{\mathcal G}

\newcommand{\calN}{\mathcal N}

\newcommand{\calV}{\mathcal V}

\newcommand{\ra}{\rightarrow}

\newcommand{\cor}{{\overline r}} 

\newcommand{\tb}{f}

\newcommand{\fw}{\text{\sc FPD}}
\newcommand{\hist}{\text{Hist}}
\newcommand{\histset}[3]{{\mathbb H}_{#1,#2}^{#3}}

\newcommand{\equal}{\text{Eq}}

\newcommand{\Omit}[1]{}

\newcommand{\calD}{\mathcal D}

\newcommand\boxit[1]{{\begin{tabular}{|@{\hspace{.5mm}}c@{\hspace{.5mm}}|}\hline$#1$\\ \hline\end{tabular}} }
\newcommand{\mn}{{\mathcal N}}

\newcommand{\approval}{\text{app}}

\newcommand{\sgroup}{{\cal S}_{\ma}}

\newcommand{\closure}[1]{\overline{#1}}

\newcommand{\CL}{\text{CL}}

\newcommand{\ano}{\text{\sc Ano}}
\newcommand{\neu}{\text{\sc Neu}}
\newcommand{\res}{\text{\sc Res}}

\newcommand{\anr}{\text{\sc ANR}}

\newcommand{\prefspace}{{\cal E}}
\newcommand{\decspace}{{\cal D}}
\newcommand{\listset}[1]{{\cal L}_{#1}}
\newcommand{\committee}[1]{{\cal M}_{#1}}
\newcommand{\commonpref}[1]{\text{\sc C}_{#1}}

\newcommand{\commondec}[1]{\text{\sc S}_{#1}}

\newcommand{\card}[1]{
\ifthenelse{\equal{#1}{}}{Empty.}{Nonempty.}
}

\newcommand{\anrposs}[1][]{\text{\sc ANR-possibility}
\ifthenelse{\equal{#1}{}}{}{\text{-}#1}
}

\newcommand{\anrcert}[1][]{\text{\sc ANR-cert}
\ifthenelse{\equal{#1}{}}{}{\text{-}#1}
}

\newcommand{\reducesto}{\leq_T^P}

\newcommand{\GI}{\text{GI}}

\newcommand{\GA}{\text{GA}}

\newcommand{\CS}[1]{\text{\sc CS}_{#1}}
\newcommand{\scsetting}{{(\prefspace,\decspace)}}

\newcommand{\ap}{\text{\bf ap}}
\newcommand{\apv}{\text{\sc APV}}

\newcommand{\cf}{\text{\bf cf}}
\newcommand{\cl}{\text{\bf cl}}
\newcommand{\cllex}{\text{\bf cl}_\text{lex}}
\newcommand{\rs}{\text{\bf rs}}

\newcommand{\merv}{\text{MERV}}
\newcommand{\merev}{\text{MEReV}}
\newcommand{\veri}{\textcircled{v}}

\newcommand{\verified}[1]{
#1_{\hspace{-.5mm}\tiny\veri}
}

\newcommand{\twolines}[2]{\begin{tabular}{@{}c@{}}#1\\#2\end{tabular}}

\newcommand{\CLR}[1][]{
\ifthenelse{\equal{#1}{\text{CLR}}}
{}{\text{CLR}_{#1}}
}

\newcommand{\CLTB}[1][]{
\ifthenelse{\equal{#1}{\text{CLTB}}}
{}{\text{CLTB}_{#1}}
}

\usepackage[most]{tcolorbox}
\newtcolorbox{KRQ}[2][]{
                lower separated=false,
                colback=white,
colframe=black,fonttitle=\bfseries,
colbacktitle=black!70,
coltitle=white,
enhanced,
attach boxed title to top left={yshift=-0.1in,xshift=0.15in},
                 boxed title style={boxrule=0pt,colframe=white,},
title=#2,#1}

\title[Computing Most Equitable Voting Rules]{Computing Most Equitable Voting Rules}
\author{
Lirong Xia}  
\affiliation{
\institution{Rutgers University-New Brunswick and DIMACS} 
\country{USA} 
}
\email{lirong.xia@rutgers.edu}

\begin{document}


\begin{abstract}  
How to design fair and (computationally) efficient voting rules is a central challenge in Computational Social Choice. In this paper, we aim at designing efficient algorithms for computing {\em most equitable rules} for large classes of preferences and decisions, which optimally satisfy two fundamental fairness/equity axioms: {\em anonymity} (every voter being treated equally) and {\em neutrality} (every alternative being treated equally). By revealing a natural connection to the graph isomorphism problem and leveraging recent breakthroughs by~Babai [2019], we design quasipolynomial-time algorithms that compute  {\em most equitable rules with verifications}, which also compute  verifications about whether anonymity and neutrality are satisfied at the input profile. Further extending this approach, we propose the {\em canonical-labeling tie-breaking}, which runs in quasipolynomial-time and optimally breaks ties to preserve anonymity and neutrality. As for the complexity lower bound, we prove that even computing verifications for most equitable rules is $\GI$-complete (i.e., as hard as the graph isomorphism problem), and sometimes $\GA$-complete (i.e., as hard as the graph automorphism problem), for many commonly studied combinations of preferences and decisions. To the best of our knowledge, these are the first problems in computational social choice that are known to be complete in the class $\GI$ or $\GA$.
\end{abstract}


\maketitle
 
\setcounter{page}{1} 
\section{Introduction}


In classical social choice settings, it is often assumed that agents' preferences are linear orders over all alternatives and the collective decision is a single winning alternative.  Motivated by the increasing public need for fair, efficient, and trustworthy collective decisions in modern society~\citep{Mancini2015:Why-it-is-time}, the applications of social choice theory have been significantly expanded to settings with a large variety of preferences and decisions~\citep{Xia2022:Group}. Modern applications often allow agents to express much richer types of preferences, such as {\em list preferences} (linear orders over a subset of alternatives, e.g., in information retrieval~\citep{Liu11:Learning}, recommender systems~\citep{Dwork01:Rank}, and crowdsourcing~\citep{Mao12:Social}) and {\em committee preferences} (sets of alternatives, e.g., in approval voting), and the collective decision can also be a list or a committee as well.

Among all  equity/fairness axioms, {\em anonymity} (all agents being treated equally)  and {\em neutrality} (all alternatives being treated equally) are broadly viewed as ``{\em minimal demands}'' and ``{\em uncontroversial}''~\citep{Myerson2013:Fundamentals,Zwicker2016:Introduction,Brandt2019:Collective}. Therefore, it is natural and desirable to design voting rules that simultaneously satisfy anonymity, neutrality, and {\em resolvability} (the rule must choose a single decision). Such rules are called {\em ANR  rules}. 
Unfortunately, no voting rule always satisfies ANR under the classical social choice settings~\citep{Moulin1983:The-Strategy}. This negative result is known as the {\em ANR impossibility}, which is ``{\em among the most well-known results in social choice theory}''~\citep{Ozkes2021:Anonymous}. Knowing that no rule can satisfy ANR for all profile, a natural goal is to design voting rules that {\em optimally} satisfies ANR, for example by maximizing the likelihood of ANR satisfaction according to some probability distributions over the profiles. But which distribution/model should be used, given that {\em ``all models are wrong''}~\citep{Box1979:Robustness}?

Surprisingly, the optimal rules are insensitive to the underlying distribution~\citep{Xia2023:Most}. Such rules are called {\em most equitable rules}, which satisfy ANR at every {\em ANR-possible profile}, where ANR can be satisfied by {\em some} rule.  In other words, a most equitable rule maximizes the likelihood of ANR satisfaction under every distribution over the profile. \citet{Xia2023:Most} proved that most equitable  rules always exist and one of them can be computed in polynomial time when agents' preferences are linear orders over all alternatives. This addresses the computational challenge discussed by~\citet{Bubboloni2021:Breaking}, but how to efficiently compute a most equitable rule for  other types of preferences, especially list preferences over a subset of alternatives and committee preferences, is left open. This is the key research question we address in this paper.

\begin{KRQ}{Key Research Question}
\begin{center}
How can we compute  most equitable rules for general preferences and decisions?
\end{center}
\end{KRQ}
Additionally, it is also important to provide a {\em verification} of ANR satisfaction, which informs the agents about the quality of the decision (i.e., whether ANR is satisfied) as well as its optimality (i.e., if any voting rule can do better). Such verifications help improve the transparency and trustworthiness of collective decision making processes, and have gained much attention from theoreticians and practitioners. For example, 
the subdomain of social choice that analyzes satisfaction of social choice axioms based on models and real-world data~\citep{Diss2021:Evaluating} can be viewed as providing verifications of desirable axioms, in particular the {\em Condorcet criterion}.  FairVote.org provides verifications of Condorcet criterion for Ranked Choice Voting to justify its usage~\footnote{\url{https://fairvote.org/resources/data-on-rcv/}}.  Surprisingly, despite the significance of ANR, we are not aware of a previous work on computing ANR verifications.

\subsection{Our Contributions}
We design efficient algorithms for {\em most equitable rules with verifications ($\merv$s)}, denoted by $\verified{r}$, for a large class of commonly studied social choice settings. Given any profile $P$, $\verified{r}(P)=(d,v)$, where $d$ is the decision and $v\in\{0,1\}$ is the verification that indicates whether ANR is satisfied at $P$. The verification can be viewed as a strong guarantee, in the sense that $v=1$ means that ANR is satisfied, while $v=0$ means that ANR cannot be satisfied at $P$ no matter what rule is used. Our overall approach is based on revealing and leveraging a natural connection between preferences and graphs. In a nutshell, we show 
$$\boxit{\GI}
\reducesto
\boxit{\text{computing  ANR verifications}}
\reducesto
\boxit{\text{computing a $\merv$}}
\reducesto
\boxit{\CL}
 $$
where $\GI$ is the {\em graph isomorphism problem}, which asks whether two given undirected unweighted graphs are isomorphic.  $\reducesto$ represents a polynomial-time Turing reduction.  $\CL$ is the {\em graph canonical labeling problem}, which computes a {\em canonical form} of a given graph. A canonical form is a graph that can be viewed as the ``representative'' of graphs that are isomorphic to each other.

\myparagraph{Social choice settings.}  We consider a large class of social choice settings $(\prefspace,\decspace)$, where $\prefspace$ and $\decspace$ denote the preference space and decision space, respectively. 
For any $m\in\mathbb N$, let  $\ma =[m] = \{1,\ldots,m\}$ denote the set of $m$ alternatives. For any  $i\le m$, let $\committee i$ denote the set of all {\em committees} over $i$ alternatives (subsets of $i$ alternatives); let $\listset i$ denote the set of all {\em lists} (linear orders) over all committees of $i$ alternatives;  and let  $\commonpref{m}\triangleq \bigcup_{i\le m} (\listset i\cup\committee i)$ denote the {\em common preferences}, which consist of all lists and committees. We call $\prefspace$ {\em common preferences}, if  $\prefspace = \commonpref{m}$; and we call  $\decspace$  {\em common decisions}, if $\decspace$ is $\listset k$ or $\committee k$ for some $k\le m$. 

\myparagraph{Easier than $\CL$.} Leveraging the connection revealed in this paper and~\citet{Babai2019:Canonical}'s quasipolynomial-time algorithm for $\CL$, our first main message is:

\mainmessage[Alg.~\ref{alg:CLVR-general}, Thm.~\ref{thm:CLVR-general}, and Coro.~\ref{coro:MERC-quasipoly}]{1}{For common preferences and common decisions, a $\merv$ can be computed in   quasipolynomial time.}

The message is quite positive, as  quasipolynomial-time algorithms, though slower than polynomial-time algorithms, are generally viewed as efficient. We further extend our methodology to design a new class of tie-breaking mechanisms called {\em canonical-labeling tie-breaking ($\CLTB$)}, which break ties to optimally satisfy anonymity and neutrality. Our second (positive) message is that such tie-breaking mechanisms with verifications can be efficiently computed as well. 

\mainmessage[Alg.~\ref{alg:CLTB-general}, Thm.~\ref{thm:CLTB-general}, and Coro.~\ref{coro:MERC-quasipoly}]{2}{For common preferences and common decisions,  a most equitable tie-breaking mechanism with verification can be computed in quasipolynomial time.}

Our approaches can further benefit from faster $\CL$ algorithms. This is possible because whether $\CL=P$ is an open question---while $\CL$ is harder than $\GI$, whether it is stricter harder is an open question, and whether $\GI = P$ is an open question. 
More discussions can be found in Section~\ref{sec:related}.

\myparagraph{Harder than $\GI$.} Because a most equitable rule can be computed in polynomial time when $m$ is bounded (Theorem~\ref{thm:bounded-m}), we consider a series of common social choice settings for variable $m$. One example is $(\committee \iota, \committee \kappa)$, where $\iota$ and $\kappa$ are functions $\mathbb N\ra\mathbb N$ such that for any  $m$, agents' preferences are $\iota(m)$ committees and the decision space are $\kappa(m)$ committees. We prove that computing verifications for $\merv$s, which is equivalent to computing  whether any rule (equivalently, any most equitable rule) satisfies ANR at any given profile, denoted by $\anrposs$, is $\GI$-complete or $\GA$-complete. This is our third main message.

\mainmessage[Table~\ref{tab:summary}]{3}{In most common settings, $\anrposs$ is $\GI$-C or $\GA$-C.}

\begin{table}[htp]
\begin{tabular}{@{}c@{ }c @{}}
\begin{NiceTabular}{|@{\ \small}wr{4.5em}@{\hspace{.1em}}|@{\small}c@{}|@{\small}c@{}|}
\cline{1-3}
\diagbox{\small $\prefspace =\committee\iota$}{\small $\decspace = \committee\kappa$}&  \twolines{\multirow{2}{*}{\small $1,\ldots,m-1$}}{\ } &\small $m$\\
\hline
 \small$1$ &  \small P (Thm.~\ref{thm:iota=1})& { }\multirow{4}{*}{\small\twolines{YES}{(Thm.~\ref{thm:m-committee})}} { }\\
\cline{1-2}
\small  $2,\ldots,m-2$ &  { }\small GI-C (Thm.~\ref{thm:GI-hardness-MM}) { }&  \\
\cline{1-2}
\small$m-1$ & \small  P (Thm.~\ref{thm:iota=1}) & \\
\cline{1-2}
\small $m$ & \small NO (Thm.~\ref{thm:m-committee}) &   \\
\hline
\end{NiceTabular}
&
\setlength{\tabcolsep}{3pt}
\begin{NiceTabular}{|@{\ }wr{4.5em}@{\hspace{.1em}}| c  |  c  |  c |}
\cline{1-4}
\diagbox{\small $\prefspace =\committee\iota$}{\small $\decspace = \listset\kappa$} &\small $m- \Omega(m^\epsilon)$& \twolines{\multirow{2}{*}{\small \twolines{\multirow{2}{*}{other}}}}{\ } &\small {$m-1$}{\& $m$}\\
\hline
 \small $1$ &   \multicolumn{3}{c|}{\small P (Thm.~\ref{thm:iota=1})}\\
 \cline{1-4}
\small  $2,\ldots,m-2$ &  \small GI-C (Thm.~\ref{thm:GI-hardness-ML}) &\small GA-H (Thm.~\ref{thm:GA-H}) & \small GA-C (Thm.~\ref{thm:GA-H})  \\
\hline
\small $m-1$ &   \multicolumn{3}{c|}{\small P (Thm.~\ref{thm:iota=1})}\\
\hline
\small $m$ &  \multicolumn{3}{c|}{\small NO (Thm.~\ref{thm:m-committee})}\\
\hline
\end{NiceTabular}
\\
\small (a) $(\prefspace,\decspace) = (\committee \iota, \committee \kappa)$ 
&
\small (b) $(\prefspace,\decspace) = (\committee \iota, \listset \kappa)$
\\
\ \\
\begin{NiceTabular}{|@{\ }wr{4.5em}@{\hspace{.1em}}|@{}c@{}|@{}c@{}|}
\cline{1-3}
\diagbox{\small $\prefspace =\listset\iota$}{\small $\decspace = \committee\kappa$}   &   \twolines{\multirow{2}{*}{\small $1,\ldots,m-1$}}{\ } & \small$m$\\
\hline
 $1$ & \small P (Thm.~\ref{thm:iota=1}) &{ }\small \multirow{4}{*}{\twolines{YES}{(Thm.~\ref{thm:m-committee})}} { }\\
\cline{1-2}
\small $m- \Omega(m^\epsilon)$&{ }\small {GI-C} {(Thm.~\ref{thm:GI-hardness-LM})}{ } &  \\
\cline{1-2}
 \small other  & ? & \\
\cline{1-2}
 \small $m-O(1)$& \small  P (Thm.~\ref{thm:const-unranked}) &   \\
\hline
\end{NiceTabular}
&
\setlength{\tabcolsep}{3pt}
\begin{NiceTabular}{|@{\ }wr{4.5em}@{\hspace{.1em}}| c | c | c |}
\cline{1-4}
\diagbox{\small $\prefspace =\listset\iota$}{\small $\decspace = \listset\kappa$} &\small $m- \Omega(m^\epsilon)$& \twolines{\multirow{2}{*}{\small \twolines{\multirow{2}{*}{other}}}}{\ } &\small  {$m-1$}{\& $m$}\\
\hline
 $1$ &   \multicolumn{3}{c|}{\small P (Thm.~\ref{thm:iota=1})}\\
 \cline{1-4}
\small $m- \Omega(m^\epsilon)$  &  \small  GI-C (Thm.~\ref{thm:GI-hardness-LL})& \small GA-H (Thm.~\ref{thm:GA-H})& \small GA-C (Thm.~\ref{thm:GA-H}) \\
\cline{1-4}
other & \multicolumn{2}{c}{?} & \small GA-H (Thm.~\ref{thm:GA-H})\\
\hline
\small $m-O(1)$&   \multicolumn{3}{c|}{\small P (Thm.~\ref{thm:const-unranked})}\\
\hline
\end{NiceTabular}
\\
\small (c) $(\prefspace,\decspace) = (\listset \iota, \committee \kappa)$ 
&
\small (d) $(\prefspace,\decspace) = (\listset \iota, \listset \kappa)$
\end{tabular}
\caption{Complexity of   $\anrposs$ for any constant $0<\epsilon\le 1$. Each of the subscripts $\iota$ and $\kappa$ represents a series of preference settings and decision settings, respectively, one for each $m$. The first column and the first row represent choices of $\iota$ and $\kappa$, respectively. ``YES'' or``NO'' means that the answer is   always ``YES'' or ``NO'', respectively. ``P'' stands for polynomial-time computable. ``GI-C'' stands for $\GI$-complete. ``GA-H'', ``GA-C'', and ``GA-E'' stand for $\GA$-hard, $\GA$-complete, and $\GA$-easy, respectively.  ``?'' stands for open question.
\label{tab:summary}}
\end{table}
As shown in  Table~\ref{tab:summary}, in all four types of combinations of preference space and decision space, $\anrposs$ is $\GI$-C  when  $\iota$ is neither too large nor too small, and $\kappa$ is not too large. For example, $\anrposs$ is $\GI$-C under $(\committee \iota, \committee\kappa)$  where $2\le \iota(m)\le m-1$ and $\kappa(m)\le m-1$ (Table~\ref{tab:summary} (a)). For $\GA$-completeness,  Table~\ref{tab:summary} (b) and (d) show that when $\iota$ is neither too large nor too small and $\kappa$ is $m-1$ or $m$, $\anrposs$ is $\GA$-C.  To the best of our knowledge, these are the first set of problems in (Computational) Social Choice, perhaps even in Economics and Computation, that are known to be $\GI$-C or $\GA$-C.


\subsection{Related Work and Discussions}
\label{sec:related}
\noindent{\bf Social choice settings.} As discussed in~\citep{Xia2023:Most}, the list/committee preferences and decisions cover many commonly studied social choice problems: $(\listset \ell,\listset k)$ as in Arrowian framework, $(\listset \ell,\committee k)$ as in multi-winner voting, $(\committee \ell,\listset k)$ as in approval voting, $(\committee \ell,\committee k)$ as in approval-based committee voting. The proposed polynomial-time most equitable rule in~\citep{Xia2023:Most} only works for $\prefspace = \listset m$. Our $\CL$-easiness results for $\CLR[\cl]$ and $\CLTB[\cl]$ consider a more general preference space that consists of all lists and committees. Our $\GI$-hardness results (computing an ANR verification) consider a different problem setting, where we are given a series of social choice settings for variable $m$. 

\myparagraph{Anonymity, neutrality, and the ANR impossibility.} Anonymity and neutrality are fundamental notions of equity  and fairness, whose formal studies can be dated back to as early as~\citet{May52:Set}'s seminal work that characterizes the majority rule. Conditions on $m$ and $n$ for the ANR impossibility to hold have been obtained for various commonly studied combinations of preferences and decisions studied in this paper: $(\prefspace,\decspace) = (\listset m,\listset k)$~\citep{Moulin1983:The-Strategy,Campbell2015:The-finer,Dogan2015:Anonymous,Ozkes2021:Anonymous,Bubboloni2014:Anonymous,Bubboloni2015:Symmetric} and  $(\prefspace,\decspace) = (\listset m,\committee k)$~\citep{Bubboloni2016:Resolute,Bubboloni2021:Breaking}. Recently, the characterizations were  generalized to all $(\prefspace,\decspace)\in \CS{m}\times \CS{m}$~\citep{Xia2023:Most}.  

\myparagraph{Tie-breaking mechanisms.} Some previous work focused on designing tie-break mechanisms to preserve as much ANR satisfaction as possible for    $(\prefspace,\decspace) = (\listset m,\listset 1)$~\citep{Dogan2015:Anonymous,Bubboloni2016:Resolute}.  \citet{Xia2020:The-Smoothed} showed that the commonly studied agenda tie-breaking and the fixed-voter tie-breaking are far from being optimal under a large class of semi-random models, proposed a tie-breaking mechanism, and proved that it is asymptotically optimal for even $n$'s. \citet{Xia2023:Most} proposed a polynomial-time computable  most equitable tie-breaking mechanism   for $\listset m$ preferences. Our $\CLTB[\cl]$ generalizes this tie-breaking mechanism to  $\listset{m-O(1)}$-preferences   (Theorem~\ref{thm:bounded-m}), and it runs in quasipolynomial time for a much larger class of preferences, i.e., $\commonpref m$.

\myparagraph{Computing voting rules and verifications.} The study of computational complexity of voting rules is a classical topic that shaped the field of Computational Social Choice~\citep{Brandt2016:Handbook}. There is a large body of previous work on investigating computational complexity and designing algorithms for commonly-studied voting rules,  for example for $\decspace = \committee 1$ (e.g., Dodgeson~\citep{Bartholdi89:Voting,HHR97}, Kemeny~\citep{Hemaspaandra05:Complexity},  Young~\citep{Rothe2003:Exact}, RCV/STV~\citep{Conitzer09:Preference}) and $\listset m$ (e.g., Kemeny ranking and Dodgeson ranking~\citep{Bartholdi89:Voting,Hemaspaandra05:Complexity}, Young ranking~\citep{Rothe2003:Exact}, and Slater~\citep{Conitzer06:Slater}) and for $\decspace = \committee k$ (e.g.,  Chamberlin-Courant and Monroes' rule~\citep{Procaccia07:Multi}, budgeted social choice~\citep{Lu2011:Budgeted}, and OWA-based rules~\citep{Skowron2016:Finding}).

Another closely related recent line of research aims to achieve proportional representation for committee voting, i.e., $\decspace = \committee k$. Much research has focused on proposing desirable axioms and computing a decision that satisfies them~\cite{Aziz2017:Justified,Cheng2020:Group,Peters2020:Proportionality,Skowron2021:Proportionality,Brill2023:Robust}, as well as computing verifications for their satisfaction~\citep{Aziz2017:Justified,Aziz2018:On-the-Complexity,Aziz2020:The-expanding,Janeczko2022:The-complexity,Brill2023:Robust}.  See~\citep{Lackner2023:Multi-Winner} for a comprehensive survey. While the high-level goal of computing winners and verifications (of satisfaction of axioms) in social choice is not new, we are not aware of previous work that computes most ANR rules for general preferences and decisions (except~\citep{Xia2023:Most}, which tackles $\prefspace = \listset m$), and we are not aware of a previous work that computers ANR verifications.
 
\myparagraph{Graph isomorphism and computation.} $\GI$ is a fundamental and well-studied problem that plays an important role in complexity theory, partly because it is {\em ``one of the few natural problems in the complexity class NP that could neither be classified as being hard (NP-complete) nor shown to be solvable with an efficient algorithm''}~\citep{Grohe2020:The-graph}.  It is also widely believed that $\GI$ is not NP-complete, because otherwise the polynomial hierarchy NP collapses to its second level, which is believed to be unlikely. It is well-known that $\GA\reducesto\GI\reducesto \CL$, yet no reduction in any reverse direction is known. There is a large body of literature on algorithms for $\GI$, which can be found in various surveys over the past 50 years, e.g., \citep{Booth1977:Problems,Zemlyachenko1985:Graph,Grohe2020:The-graph}. The state of the art is the breakthrough of the quasipolynomial-time algorithm for $\GI$ by~\citet{Babai2018:Group}, which was later extended to a quasipolynomial-time algorithm for $\CL$ by~\citet{Babai2019:Canonical}. The positive messages of our paper (easiness of computing $\merv$ and  most equitable tie-breaking mechanism with verification) leverage the latter work, and they can benefit from any improvement in $\CL$ algorithms in light of the connections revealed in this paper.  
We believe that computing $\merv$ and computing $\anrposs$ are interesting to the theory community as well because it is a natural problem to develop efficient algorithms and build complexity theories.

 \section{Preliminaries}
\label{sec:prelim}

\myparagraph{\bf Social choice setting $(\prefspace,\decspace)$.} Let $\ma=[m] = \{1,\ldots, m\}$ denote the set of $m\ge 2$ {\em alternatives}.  There are  $n\in\mathbb N$ agents, each using an element in the {\em preference space} $\prefspace$ to represent his or her preferences, called a {\em vote}.  For any  $i\le m$, let $\listset i$ denote the set of all {\em $i$-lists}  (linear orders over $i$ alternatives) and let $\committee i$ denote the set of all {\em $i$-committees} (subsets of $i$ alternatives).  Let the {\em common preferences} $\commonpref{m}$ denote the set of all $\ell$-lists and $\ell$-committees over $\ma$, that is, $\commonpref m\triangleq \bigcup_{\ell\le m} (\listset \ell\cup\committee \ell)$.  We also define a lexicographic order $1\rhd\cdots\rhd m$ and extend it to $\commonpref m$ in a natural way, such that elements within $\committee {i}$ and elements within $\listset {i}$ are compared lexicographically, and the former have higher priorities than the latter; and for all $i_1<i_2$, elements in $\committee {i_1}$ (respectively, $\listset {i_1}$) have higher priority than elements in $\committee {i_2}$ (respectively, $\listset {i_2}$).

Given  $\prefspace$, the vector of $n$ agents' votes, denoted by $P$, is called a {\em (preference) profile}, sometimes called an $\prefspace$-profile. Let $\calD$ denote the decision space. Let $\commondec m\triangleq\{\listset{i},\committee{i}:1\le i\le m\}$. We say that $\calD$ is a {\em common decision}, if  $\decspace\in \commondec m$. For any $m\in\mathbb N$, the {\em Common Setting}~\citep{Xia2023:Most}, is defined as $\CS{m}\triangleq \commondec m \times \commondec m$. Notice  $\commondec m$ consists of {\em sets} of preferences,  while $\commonpref{m}$ consists of preferences.

\myparagraph{Histograms.} Let $\hist(P)\in {\mathbb Z}_{\ge 0}^{|\prefspace|}$ denote the {\em histogram} of $P$, which can be viewed as the anonymized $P$ that contains the number of occurrences of each type of preferences in  $P$.  We sometimes represent a histogram $\vec h$ in a compact format $(h_1\times R_1,\ldots, h_{n'}\times R_{n'})$, where $h_j$'s are positive integers such that $\sum_{j=1}^{n'}h_j = n$, and $R_j\in \prefspace$ are sorted in lexicographic order. Let $\histset{m}{n}{\prefspace}$ denote the set of all histograms of  $\prefspace$-profiles of $n$ votes. The lexicographic priority order $\rhd$ can be naturally extended to $\histset{m}{n}{\commonpref m}$.


\myparagraph{\bf Voting rules.} Given a social choice setting $(\prefspace,\decspace)$, an {\em  irresolute} voting rule $\cor: \prefspace^n\ra(2^\decspace\setminus\emptyset)$ is a
mapping from an $\prefspace$-profile to a non-empty subset of $\decspace$. For example, when $(\prefspace,\decspace)=(\commonpref m,\listset 1)$, for any $t\in [m]$ let $\closure{\approval_t}$  denote the irresolute  {\em top-$t$-approval rule}  such that for every $R\in \committee \ell$, every alternative in $R$ receives $1$ point, and for every $R\in \listset \ell$, only the alternatives ranked within top $t$ in $R$ receive $1$ point. The winners are the alternatives with maximum total points.   
 A {\em resolute}  voting rule $r: \prefspace^n\ra \decspace$ is an irresolute rule that always chooses a single decision. We slightly abuse the notation by using $r(P)=d$ and $r(P)= \{d\}$ interchangeably. A voting rule $\cor'$ is a {\em refinement} of another voting rule $\cor$ if for all profiles $P$, $\cor'(P)\subseteq \cor(P)$.   
 
\myparagraph{\bf Tie-breaking mechanisms.}  Many commonly studied resolute voting rules are defined as the result of a {\em tie-breaking mechanism}   applied to the outcome of an irresolute rule. A tie-breaking mechanism $\tb$ is a mapping from a profile $P$ and a non-empty set $D\subseteq \calD$, which represents co-winners, to a single decision in $D$.  For example,   the {\em lexicographic tie-breaking} breaks ties in favor of alternatives ranked higher w.r.t.~a pre-defined   priority order $\rhd$ over $\decspace$. 
Let $\tb\ast\cor$ denote the refinement of $\cor$ by applying $f$. That is, for any profile $P$, $(\tb\ast\cor)(P) =  \tb(P,\cor(P)) $. 

\myparagraph{Permutations.} Let $\sgroup$ denote the set of all permutations over $\ma$. A permutation can be represented by its cycle form. For example, when $m=4$, $\sigma= (1,2)(3,4)$ denote the permutation that exchanges $1$ and $2$, and exchanges $3$ and $4$.   Any permutation $\sigma\in \sgroup$ can be naturally extended to  $k$-committees, $k$-lists, profiles, and histograms over $\ma$ as follows. For any $k$-committee $M = \{a_1,\ldots,a_k\}$, let $\sigma(M)\triangleq \{\sigma(a_1),\ldots,\sigma(a_k)\}$; for any $k$-list $R = [a_1\succ \cdots \succ a_k]$, let $\sigma(R)\triangleq [\sigma(a_1)\succ \cdots\succ\sigma(a_k)]$; for any profile $P = (R_1,\ldots, R_n)$, let $\sigma(P) \triangleq (\sigma(R_1),\ldots, \sigma(R_n))$; 
and for any histogram $\vec h \in \histset{m}{n}{\prefspace}$, let $\sigma(\vec h)$ be the histogram such that for every ranking $R$, $[\sigma(\vec h)]_R = [\vec h]_{\sigma^{-1}(R)}$, where $[\sigma(\vec h)]_R$ is the value of $R$-component in $\sigma(\vec h)$. For any pair of histograms $\vec h_1,\vec h_2$, if $\sigma(\vec h_1)=\vec h_2$, then we say that $\vec h_1$ is {\em isomorphic} to $\vec h_2$, denoted by $\vec h_1\cong\vec h_2$, and $\sigma$ is called an {\em isomorphism} between $\vec h_1$ and $\vec h_2$.

\noindent\begin{minipage}[t][][b]{.6\textwidth}
\begin{ex}
\label{ex:perm}
Let $m=4$, $n=4$,  and $P_1$ and $P_2$ denote any $\committee 2$-profiles whose histograms $\vec h_1$ and $\vec h_2$ are indicated in  table on the right. Let $\sigma= (1,2)(3,4)$. Then $\sigma(\vec h_1)=\vec h_2$, i.e., $\vec h_1\cong\vec h_2$, and $\sigma$ is an isomorphism  between them.
\end{ex}
\end{minipage}
\begin{minipage}[t][][b]{.41\textwidth}
\centering 
\begin{tabular}{|@{}c@{}|@{}c@{}|@{}c@{}|@{}c@{}|@{}c@{}|@{}c@{}|@{}c@{}|}
\hline $\committee 2$& $\{1,2\}$& $\{1,3\}$& $\{1,4\}$ &$\{2,3\}$ & $\{2,4\}$ & $\{3,4\}$ \\
\hline $\vec h_1$ & $1$&$0$& $1$& $1$& $1$&$0$  \\
\hline $\vec h_2$& $1$&$1$& $1$& $1$& $0$&$0$  \\
\hline
\end{tabular}
\end{minipage}

\myparagraph{\bf Anonymity, neutrality, and resolvability.} 
For any rule $\cor$ and any profile $P$, we define $\ano(\cor,P) \triangleq  1$ if  for any profile $P'$ with $\hist(P')=\hist(P)$,   $\cor(P')=\cor(P)$; otherwise $\ano(\cor,P) \triangleq  0$. We define $\neu(\cor,P) \triangleq  1$ if   for every permutation $\sigma$ over $\ma$, we have $\cor(\sigma(P))=\sigma(\cor(P))$; otherwise $\neu(\cor,P) \triangleq 0$. We define $\res(\cor,P) \triangleq  1$ if $|\cor(P)|=1$; otherwise $\res(\cor,P) \triangleq  0$. Notice that a resolute rule outputs a single {\em decision}, which may not be a single alternative---for example when $\decspace = \committee 2$, a decision is a set of two alternatives. If $\ano(\cor,P) = 1$ (respectively, $\neu(\cor,P) = 1$ or $\res(\cor,P) = 1$), then we say that $\cor$ satisfies anonymity (respectively, neutrality or resolvability) at $P$. We further define $\anr(\cor,P)  \triangleq \ano(\cor,P)\times \neu(\cor,P)\times \res(\cor,P)$. That is, $\anr(\cor,P)=1$ ($\cor$ satisfies {\em ANR} at $P$) if and only if $\cor$ satisfies all three properties at $P$.  Next, we recall the definition of ANR-possible profiles and most equitable rules~\citep{Xia2023:Most}. 
\begin{dfn}[{\bf ANR-possible profiles and most equitable rules}{}]
\label{dfn:anr-possible}
For any social choice setting, a profile $P$ is  {\em ANR-possible}, if there exists a voting rule $r^*$ such that $\anr(r^*,P)=1$. A resolute rule $r$ is {\em most equitable}, if it satisfies ANR at all ANR-possible profiles.
\end{dfn}
In other words, if a profile $P$ is not ANR-possible, then no voting rule can satisfy ANR at $P$. Any most equitable rule (optimally) satisfies ANR at every profile.

\myparagraph{Canonical form and canonical labeling.}
A {\em canonical form} for graphs is a mapping $\cf$ from graphs to graphs such that (1) for every graph $G$, $\cf(G)$ is isomorphic to $G$, denoted by $\cf(G)\cong G$, and (2) for any pairs of graphs $G$ and $G'$, $\cf(G)=\cf(G')$ if and only if $G\cong G'$. Given a $\cf$, a {\em canonical labeling} is a function $\cl$ that maps each graph $G$ to a permutation of its vertices that brings $G$ to $\cf(G)$. 
In this paper, we will make use of the lexicographic canonical labeling defined as follows.

\begin{dfn}[{\bf Lexicographic canonical labeling}{}]
Let $\cllex$ denote the {\em lexicographic canonical labeling} that maps any graph $G$ over $\ma$ to a graph that has the highest priority order (among all graphs that are isomorphic to $G$), where the set of edges are compared w.r.t.~the natural lexicographic extension of order $1\rhd2\rhd\cdots\rhd m$. We write $G_1\rhd G_2$, if $G_1$ has a higher priority than $G_2$.
\end{dfn}
For example, when $m=4$,  $\{1,2\}\rhd \{1,3\}\rhd \{1,4\}\rhd \{2,3\}\rhd \{2,4\}\rhd \{3,4\}$.

\noindent\begin{minipage}[t][][b]{.6\textwidth}
\begin{ex}[{\bf Lexicographic canonical labeling}{}]
\label{ex:graph}   
Let $G_1$ and $G_2$ denote the graphs in the figure on the right. 
We have $G_2\rhd G_1$, because both graphs contain the top-priority edge $\{1,2\}$, and only $G_2$ contains the second-priority edge $\{1,3\}$. We also have $G_1\cong G_2$. In fact, $\cllex(G_1)=G_2$.
\end{ex}
\end{minipage}
\ \ \ 
\begin{minipage}[t][][b]{.35\textwidth}
\includegraphics[width=1\textwidth]{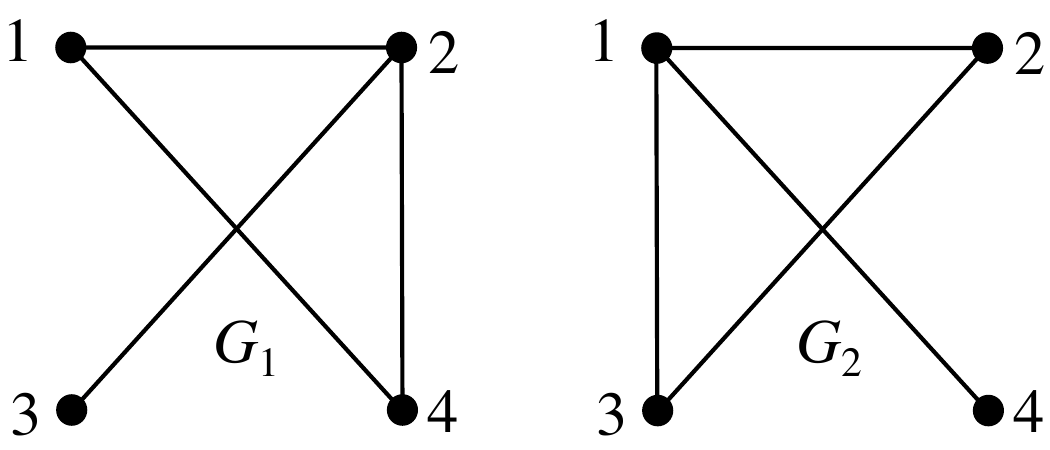} 
\centering
\end{minipage}

\begin{dfn}[{\bf Computational problems: $\GI$, $\GA$, and $\CL$}{}]
\label{dfn:GI}
In the {\em graph isomorphism ($\GI$)}  problem, the inputs are two undirected unweighted graphs $G_1,G_2$, and the output is YES if $G_1$ and $G_2$ are isomorphic; and is NO otherwise. In the {\em graph automorphism ($\GA$)}  problem, the input is a graph $G$ and the output is YES if and only if there is a non-trivial automorphism, which is an isomorphism from $G$ to itself  such that at least one vertex is mapped to different vertex.  In the {\em graph canonical labeling  ($\CL$)} problem for a given canonical form, the input is a graph $G$ and the output is a permutation that maps $G$ to its canonical form.
\end{dfn}

\myparagraph{Organization of the paper.} In Section~\ref{sec:characterizaion-MER}, we convert the problem of computing a most equitable rule to two computational problems, which  will be later shown to be closely related to $\CL$ and $\GI$, respectively. In Section~\ref{sec:simple-case}, we focus on the $\scsetting=(\committee 2,\committee 1)$ case to give the reader a taste of our approach. Due to the space constraint, the more  general and challenging cases will be tackled in Appendix~\ref{sec:general-easy} ($\CL$-easiness) and Appendix~\ref{sec:general-hard} ($\GA$/$\GI$-hardness). 

\section{Characterization of Most Equitable Rule with Verification}
\label{sec:characterizaion-MER}
We first recall a few formal definitions to present the characterization.

\begin{dfn}[{\bf Most equitable rule with verification ($\merv$)}{}] 
Given $(\prefspace,\decspace)$, a {\em {\bf M}ost {\bf E}quitable  {\bf R}ule with {\bf V}eification  ($\merv$)}  for $n$ voters is a mapping
$\verified{r}: \prefspace^n\ra \decspace\times \{0,1\}$ that consists of  two functions $\langle r,\veri \rangle$, such that (1) for any $\prefspace$-profile $P$, $\verified{r}(P) = (r(P),\veri(P))$, (2) $r$ is a most equitable rule, and (3) $\veri$ is a {\em verification function} such that for any profile $P$, $\veri(P)=1$ if and only if $P$ is ANR-possible. 
\end{dfn}

\begin{dfn}[{\bf Representative selection function}{}]
\label{dfn:rs}
For $\commonpref m$-histograms, a {\em representative selection function} is a mapping $\rs: \histset{m}{n}{\commonpref m}  \ra \sgroup$ such that for any pair of isomorphic histograms $\vec h_1\cong \vec h_2$, we have $[\rs(\vec h_1)](\vec h_1) = [\rs(\vec h_2)](\vec h_1)$.  For any $\commonpref m$-profile $P$, we define $\rs(P)\triangleq  \rs(\hist(P))$.
\end{dfn}
That is, for any histogram $\vec h$, $\rs$ outputs a permutation over $\ma$, denoted by $\sigma$, that maps $\vec h$ to a ``representative'' histogram. This enables a consistent way to make decisions for profiles whose histograms are isomorphic to each other. Still, we need to select a decision at a representative histogram to satisfy ANR if possible. Such decisions are called {\em fixed-point decisions}~\citep{Xia2023:Most}.

\begin{dfn}[{\bf Fixed-point decisions for $\commonpref m$}{}]
\label{dfn:fixed-point-decisions}
For any  $\decspace\in\commondec m$, a decision $d\in\decspace$ is a {\em fixed-point decision} of a $\commonpref m$-histogram $\vec h$, if for every permutation $\sigma$ such that $\sigma(\vec h)=\vec h$, we have $\sigma(d) = d$. Let $\fw_{\decspace}(\vec h)$ denote the set of all fixed-point decisions. 
\end{dfn}

\begin{ex}
\label{ex:fpd}
Consider the setting $(\committee 2,\committee 1)$ in Example~\ref{ex:perm}. Let $\sigma'\triangleq (1,4)$, which exchanges $1$ and $4$. It is not hard to verify that $\sigma'$ is the only permutation that maps $\vec  h_1$ to itself. Therefore, $\fw_{\decspace}(\vec h_1) = \{2,3\}$.
\end{ex}

It is known that a profile $P$ is ANR-possible if and only if $\fw_{\decspace}(\hist(P))\ne\emptyset$, and in order to satisfy ANR at $P$, any most equitable rule must choose a fixed-point decision~\citep[Lemma~4]{Xia2023:Most}. The next proposition is a direct corollary on~\citep[Lemma~4]{Xia2023:Most} to $\commonpref m$-profiles.

\begin{prop}[{\bf Characterization: Most equitable  rules}{}]
\label{prop:chara-mevr}
A voting rule $r$ for $\commonpref m$ preferences is most equitable if and only if there exists a representative selection function $\rs$ and a voting rule $r^*$ such that for every ANR-possible profile $P$, let $\sigma  \triangleq \rs(P)$, then 
$$\text{(1) } r^*(\sigma(P))\in \fw_{\decspace}(\sigma(\hist(P)))\text{, and (2) }r(P) = \sigma^{-1}(r^*(\sigma(P)))$$
\end{prop}
Proposition~\ref{prop:chara-mevr} implies that a most equitable rule can be computed in the following steps. First, use $\rs$ to convert $\hist(P)$ to a ``representative'' histogram $\sigma(\hist(P))$, then use $r^*$ to choose a fixed-point decision   if one exists, and finally map the decision back to $P$ by  $\sigma^{-1}$. In fact,  the second step computes an ANR verification as well---$P$ is ANR-possible if and only if there exists a fixed-point decision at $P^*$. We next show that computing fixed-point decisions can be reduced to computing the {\em automorphism partition}, defined as follows.

\begin{dfn}[{\bf\boldmath Automorphism partition for $\commonpref m$}{}]
\label{dfn:ap-alternatives}For any ${\commonpref m}$-profile $P$ (respectively, ${\commonpref m}$-histogram $\vec h$), its {\em automorphism partition}, denoted by  $\ap(P)$ (respectively, $\ap(\vec h)$), is a partition of $\ma$ such that  $a,b$ are in the same set if and only if there exists an automorphism $\sigma$ of $\hist(P)$ (respectively, $\vec h$) such that $b = \sigma(a)$.
\end{dfn}

We say that a set $D$ in a partition of $\ma$ is a {\em singleton}, if $|D|=1$. The following two propositions characterize fixed-point decisions and ANR-possible profiles by properties of $\ap$.
\begin{prop}[{\bf Characterization: Fixed-point decisions}{}]
\label{prop:verification}
For any  $\prefspace\in\commondec m\cup \{\commonpref m\}$, any $\decspace\in\commondec m$, and any $\prefspace$-profile $P$,  a decision $ d\in \decspace$ is a fixed-point decision if and only if
\begin{itemize}
\item [\bf (i)] {\bf when $\decspace = \listset k$:} alternatives in $d$ are singletons in $\ap(P)$;
\item [\bf (ii)] {\bf when $\decspace = \committee k$:} $d$ is the union of some sets in $\ap(P)$;
\end{itemize}
\end{prop} 
\begin{proof}
When $\decspace = \listset k$, a fixed-point decision (which is a ranking over a $k$-committee denoted by $A\subseteq \ma$) exists if and only if no permutation $\sigma$ in the stabilizer of $\hist(P)$ maps any alternative in $A$ to another alternative. This corresponds to the condition for $\decspace = \listset k$ in the statement of the proposition. When $\decspace = \listset k$, a fixed-point decision (which is a $k$-committee $A\subseteq \ma$) exists if and only if no permutation $\sigma$ in the stabilizer of $\hist(P)$ maps any alternative in $A$ to an alternative that is not in $A$. This corresponds to the condition for $\decspace = \committee k$ in the statement of the proposition.
\end{proof}

\begin{prop}[{\bf Characterization: ANR-possible profiles}{}]
\label{prop:anr-possibility}
For any  $\prefspace\in\commondec m\cup \{\commonpref m\}$, any $\decspace\in\commondec m$,  an $\prefspace$-profile $P$ is ANR-possible  if and only if
\begin{itemize}
\item [\bf (i)] {\bf when $\decspace = \listset k$:} $\ap(P)$ contains at least $k$ singletons;
\item [\bf (ii)] {\bf when $\decspace = \committee k$:} the union of some sets in $\ap(P)$ contain exactly $k$ alternatives;
\end{itemize}
\end{prop}

\begin{proof}
The proposition follows after the Lemma for general preferences~\citep{Xia2023:Most}, which states that a profile is ANR-possible, if there exists a {\em fixed-point decision}, defined as follows.  

When $\decspace = \listset k$, a fixed-point decision (which is a ranking over $k$ alternatives, denoted by $A$) exists if and only if no permutation $\sigma$ in the stabilizer of $\hist(P^*)$ maps any alternative in $A$ to another alternative. This corresponds to the condition for $\decspace = \listset k$ in the statement of the lemma. When $\decspace = \listset k$, a fixed-point decision (which is a size-$k$ committee $A\subseteq \ma$) exists if and only if no permutation $\sigma$ in the stabilizer of $\hist(P^*)$ maps any alternative in $A$ to an alternative that is not in $A$. This corresponds to the condition for $\decspace = \committee k$ in the statement of the lemma.
\end{proof}

\begin{ex}
\label{ex:ap-fd}
In the setting of Example~\ref{ex:fpd}, we have $\ap(P_1)= \ap(\vec h_1) = \{\{1,4\}, \{2\}, \{ 3\}\}$. Because $\committee 1= \listset 1$, we can apply Proposition~\ref{prop:verification} (i) or (ii), which implies that $\fw_{\decspace}(\vec h_1) = \{2,3\}$. If $\decspace = \committee 2$, then $\fw_{\decspace}(\vec h_1) = \{\{1,4\},\{2,3\}\}$ according to Proposition~\ref{prop:verification} (ii). If $\decspace = \listset 3$, then $\fw_{\decspace}(\vec h_1) = \emptyset$ according to Proposition~\ref{prop:verification} (i), which means that $P_1$  is not ANR-possible. Alternatively, Proposition~\ref{prop:anr-possibility} (i) also implies that $P_1$  is not ANR-possible.
\end{ex}
Concluding this section, we have 
the following key observation.

\begin{KRQ}{Key Observation of Section~\ref{sec:characterizaion-MER}}
\begin{center}
A most equitable rule with verification can be computed by constructing a representative selection (Def.~\ref{dfn:rs}) and an automorphism partition (Def.~\ref{dfn:ap-alternatives}).
\end{center}
\end{KRQ}

\section{The $(\committee 2, \committee 1)$ Setting: The Connection, $\CL$-Easiness, and $\GI$-Completeness}
\label{sec:simple-case}
For the purpose of illustration, in this section we  further assume that all $\committee 2$ profiles and histograms are {\em simple}, in the sense that each $2$-committee appears at most once in the profiles/histograms.


\subsection{The Connection:  $\committee 2$-Histograms $\Leftrightarrow$ Graphs}

\begin{dfn}[{\bf  $\committee 2$-histograms $\Leftrightarrow$  graphs}]
\label{dfn:M2-hist-graph} For any $\committee 2$-histogram $\vec h$, let $G_{\vec h}$ denote the graph in which the vertices are $\ma$, and there is an edge  $\{a,b\}$ for each $\{a,b\}\in \vec h$.  For any   graph $G= (V,E)$, let $\vec h_{G}$ denote the histogram whose alternatives are $V$ and for each $\{a,b\}\in E$, there is a vote $\{a,b\}$ in $\vec h$.
\end{dfn}

 \begin{ex}
 \label{ex:connection-M2}
The histograms of $P_1$ and $P_2$  in Example~\ref{ex:perm} correspond  to $G_1$ and $G_2$  in Example~\ref{ex:graph}, respectively. That is, $G_1 = \calG_{\vec h_1}$, $G_2 = \calG_{\vec h_2}$, $\vec h_1 = \vec h_{G_1}$, and $\vec h_2 = \vec h_{G_2}$. 
 \end{ex}

In Lemma~\ref{lem:simple-M2} below, we establish two connections between $\committee 2$-histograms and graphs: the first is between automorphism partitions for histograms (Definition~\ref{dfn:ap-alternatives}) and automorphism partitions for graphs (Definition~\ref{dfn:ap-graphs} below), and the second is between representative selection function for histograms (Definition~\ref{dfn:rs}) and canonical labeling for graphs. The proof follows immediately after the definitions.
\begin{dfn}[{\bf\boldmath Automorphism partition for graphs}{}]
\label{dfn:ap-graphs}
Given an unweighted graph $G = (V,E)$, its {\em automorphism partition}, denoted by $\ap(G)$, is a partition of $V$ such that $a,b$ are in the same set if and only if there exists an automorphism $\mu$ of $G$ such that $b= \mu(a)$.
\end{dfn}

\begin{lem}[{\bf The connection:  $(\committee 2,\committee 1)$ setting
}{}]  
\label{lem:simple-M2}
For any simple $\committee 2$-histogram $\vec h$, any undirected unweighted graph $G$, any representative selection function $\rs$, and any canonical labeling $\cl$, 
\begin{itemize}
\item [\bf (i) ] {\bf\boldmath  $\ap\Leftrightarrow \ap$}: $\ap(G_{\vec h}) = \ap(\vec h)$ and $\ap(\vec h_{G}) = \ap(G)$;
\item [\bf (ii)] {\bf \boldmath  $\rs\Leftrightarrow \cl$}: $\rs$ is a canonical labeling for graphs and  $\cl$ is a representative selection function for simple $\committee 2$-histograms.
\end{itemize}
\end{lem}


\subsection{Computing a Most Equitable Rule with Verification}
Recall from the Key Observation of Section~\ref{sec:characterizaion-MER} that to efficiently compute a $\merv$, it suffices to efficiently compute   $\rs$ and $\ap$ for histograms. In light of the connection revealed in Lemma~\ref{lem:simple-M2}, the two functions can be obtained by first converting the profile to a graph, and then applying $\cl$ and $\ap$ for graphs, respectively. In fact, because graph automorphism partition is $\GI$-complete~\citep{Booth1977:Problems} and $\GI\reducesto\CL$, $\ap$ can be computed by polynomially many calls to $\cl$. This leads to our   {\em canonical-labeling rule ($\CLR[\cl]$)} in Algorithm~\ref{alg:CLVR-M2}, which is parameterized by a canonical labeling algorithm $\cl$.  Notice that when $D^*=\emptyset$ (step~\ref{step:return-0}), $P$ is not ANR-possible, which means an arbitrary decision can be returned. Otherwise,  step~\ref{step:comp-d*} chooses a fixed-point decision of the ``representative'' histogram $\sigma(\vec h)$ with the highest priority according to $\rhd$.

\begin{algorithm}[htp]
\caption{{\bf  ($\CLR[\cl]$: simple $\committee 2$-profiles)} Given  a simple  $\committee 2$-profile $P$, return  $(d,c)\in\committee 1\times \{0,1\}$.\label{alg:CLVR-M2}}
\begin{algorithmic}[1]  
\STATE Compute $\vec h \triangleq \hist(P)$ and $G_{\vec h}$ according to Definition~\ref{dfn:M2-hist-graph}. 
\label{step:clr-m2-1}
\STATE Compute $\ap(G_{\vec h})$ using $\cl$. Let $D^*$ denote the set of singletons in $\ap(G_{\hist(P)})$.
\label{step:clr-m2-2}
\STATE Compute $\sigma \triangleq \cl(\hist(P))$.
\label{step:comp-d*}
\STATE  {\bf if }$D^*=\emptyset$ {\bf then return }$(\arg\max_{d\in \decspace}^{\rhd}\sigma (d), 0)$.
\label{step:return-0}
\STATE {\bf else return }$(\arg\max_{d\in D^*}^{\rhd}\sigma (d), 1)$. 
\end{algorithmic}
\end{algorithm}

\begin{ex}
In the setting of Examples~\ref{ex:perm}--\ref{ex:fpd}, let $P_1$ be the input to Algorithm~\ref{alg:CLVR-M2} with $\cl=\cllex$. Then, step~\ref{step:clr-m2-1} computes $\vec h_1$ (Examples~\ref{ex:perm}) and $G_1$ (Examples~\ref{ex:graph} and~\ref{ex:connection-M2}). Step~\ref{step:clr-m2-2} computes $\ap(G_1)=\{\{1,4\}, \{2\}, \{ 3\}\}$ and $D^* = \{2,3\}$ (Example~\ref{ex:ap-fd}). Step~\ref{step:comp-d*} computes $\sigma = (1,2)(3,4)$ (Examples~\ref{ex:perm}). Finally, since $D^*\ne\emptyset$ and $\sigma(2) = 1\rhd 4 = \sigma(3)$, the algorithm returns $(2,1)$.
\end{ex}

\begin{customthm}{\ref{thm:CLVR-general}$^-$}[{\bf\boldmath  $(\committee 2,\committee 1)$ case of Theorem~\ref{thm:CLVR-general}}{}]
\label{thm:CLVR-M2}
Under simple $\committee 2$-preferences and $\committee 1$-decisions, for any canonical labeling function $\cl$, $\CLR[\cl]$ is a $\merv$ with polynomially many calls to $\cl$.
\end{customthm}
The proof naturally follows after combining Proposition~\ref{prop:chara-mevr}, Proposition~\ref{prop:verification}, and Lemma~\ref{lem:simple-M2}, and is therefore omitted. Using~\citet{Babai2019:Canonical}'s quasipolynomial canonical labeling algorithm as $\cl$,  $\CLR[\cl]$ can be computed in quasipolynomial time as well. This also holds for more general preferences and decisions, as will be formally stated in Corollary~\ref{coro:MERC-quasipoly} in Appendix~\ref{sec:CLVR-general}.

 \subsection{Computing a Most Equitable Tie-Breaking with Verification}

In this subsection, we extend the idea behind $\CLR[\cl]$ to define tie-breaking mechanisms that optimally preserve $\anr$ for any given irresolute rule that satisfies anonymity and neutrality. The outcome after tie-breaking is called a {\em most equitable refinement}~\citep{Xia2023:Most}, and we aim at designing a  tie-breaking mechanism that computes an ANR verification as well. Recall that  an irresolute voting rule $\cor_1$ is a {\em refinement} of another irresolute voting rule $\cor_2$ if for all profiles $P$, $\cor_1(P)\subseteq \cor_2(P)$.

\begin{dfn}[{\bf Most equitable refinement with verification (\merev)}{}]
A {\em most equitable refinement}  of  an anonymous and neutral irresolute voting rule $\cor$ is a resolute voting rule $r$ that refines $\cor$, such that for every profile $P$, if there exists a refinement $r'$ of $\cor$ that satisfies ANR at $P$, then $\anr(r,P)=1$. A {\em {\bf M}ost {\bf E}quitable {\bf Re}finement with {\bf V}erification ($\merev$)}  of $\cor$ is denoted by $\verified{r} = \langle r,\veri \rangle$, where $r$ is a most equitable refinement of $\cor$ and $\veri(P)=1$ if and only if $\anr(r,P)=1$; otherwise $\veri(P)=0$. 
\end{dfn}
That is, $\veri$ verifies whether ANR is satisfied by the most equitable refinement. If $\veri(P)=0$, then no refinement of $\cor$ can satisfy ANR at $P$. Next, we present   {\em canonical-labeling tie-breaking ($\CLTB[\cl]$)} in Algorithm~\ref{alg:CLTB-M2}. The main difference with $\CLR[\cl]$ (Algorithm~\ref{alg:CLVR-M2}) is that  a fixed-point decision must be chosen from the co-winners  $D$.

\begin{algorithm}[htp]
\caption{{\bf  ($\CLTB[\cl]$: simple $\committee 2$-profiles)}  Given a  simple $\committee 2$-profile $P$ and  $D\subseteq\decspace$, return $(d,c)\in D \times \{0,1\}$.
\label{alg:CLTB-M2}}
\begin{algorithmic}[1]  
\STATE Compute $\vec h \triangleq \hist(P)$ and $G_{\vec h}$ according to Definition~\ref{dfn:M2-hist-graph}.
\label{step:cltb-m2-1}
\STATE Compute $\ap(G_{\vec h})$ using $\cl$. Let $D^*\triangleq D\cap \{\text{singletons in } \ap(G_{\vec h})\}$.
\label{step:cltb-m2-2}
\STATE Compute $\sigma \triangleq \cl(\vec h)$.
\label{step:cltb-m2-3}
\STATE  {\bf if }$D^*=\emptyset$ {\bf then return }$(\arg\max_{d\in D}^{\rhd}\sigma (d), 0)$.
\label{step:cltb-m2-5}
\STATE {\bf else return }$(\arg\max_{d\in D^*}^\rhd \sigma (d), 1)$. 
\end{algorithmic}
\end{algorithm}

\begin{ex}
We consider the $(\committee 2,\committee 1)$ setting under the top-$2$-approval rule ($\cor=\closure{\approval_2}$). Let 
$P_1$ and $P_2$ denote any pair of profiles whose histograms correspond to $G_1$ and $G_2$ in the figure below,\\
\noindent\begin{minipage}{.6\textwidth}
according to Definition~\ref{dfn:M2-hist-graph}.  Let $D\triangleq \closure{\approval_2}(P_1)= \{1,3,4,5\}$ and let $P_1$ and $D$ be the input to Algorithm~\ref{alg:CLTB-M2} with $\cl=\cllex$. Then, step~\ref{step:cltb-m2-2} computes $\ap(G_1) = \{\{1,5\},\{2,6\},\{3\},\{4\}\}$, so $D^* = \{3,4\}$. In step~\ref{step:cltb-m2-3}, $\sigma = (1,2,5,3)$, which maps $G_1$ to $G_2$.
 Because  $\sigma(3)=1\rhd 4 =\sigma(4)$, the output is $(3,1)$.\end{minipage}
\ \ \ \begin{minipage}{.36\textwidth}
\includegraphics[width=1\textwidth]{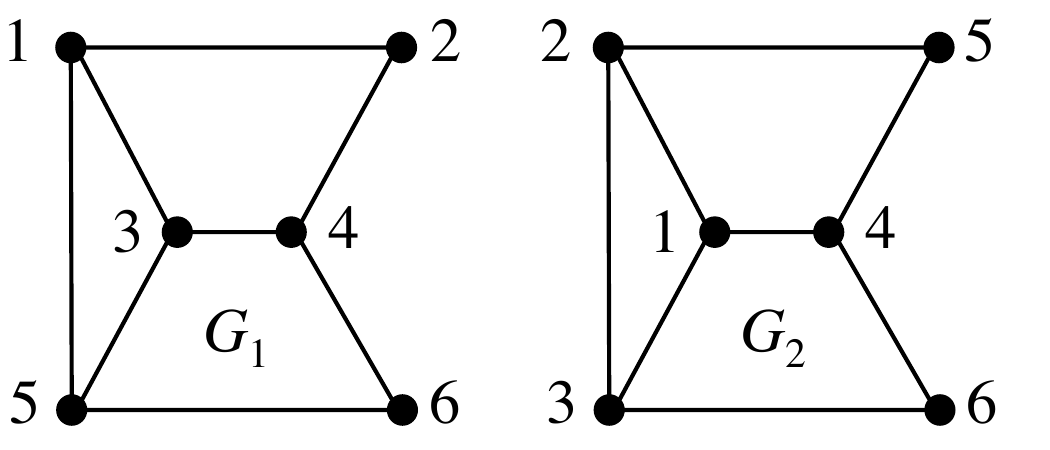} 
\end{minipage}
\end{ex}

\begin{customthm}{\ref{thm:CLTB-general}$^-$}[{\bf\boldmath  The $(\committee 2,\committee 1)$  case of Theorem~\ref{thm:CLTB-general}}{}] 
\label{thm:CLTB-M2}
Under simple $\committee 2$ preferences and $\committee 1$ decisions, for any canonical labeling function $\cl$, $\CLTB[\cl]$ computes a $\merev$ for any anonymous and neutral rule with polynomially many calls to $\cl$.
\end{customthm}
Again, the proof naturally follows after combining Proposition~\ref{prop:chara-mevr}, Proposition~\ref{prop:verification}, and Lemma~\ref{lem:simple-M2}, and by noticing that a most equitable refinement must choose a fixed-point decision from the set of co-winners~\citep[Lemma~4]{Xia2023:Most}.


\subsection{Computing an ANR Verification for Most Equitable Rules}
\label{sec:anr-cert-simple}
Recall that most equitable rules optimally satisfy ANR at every profile. Therefore, computing ANR verifications for most equitable rules is equivalent to computing whether a profile is ANR-possible (Def.~\ref{dfn:anr-possible}), formally defined as follows.
\begin{dfn}
Given a social choice setting $(\prefspace,\decspace)$, in the $\anrposs[(\prefspace,\decspace)]$ problem ($\anrposs$ for short), we are given an $\prefspace$-profile $P$, and the output is $1$ (YES) if $P$ is ANR-possible; otherwise the output is $0$ (NO).
\end{dfn}
According to Lemma~\ref{lem:simple-M2} (i) and Proposition~\ref{prop:anr-possibility} (i) or (ii) (because $\committee 1 = \listset 1$), and because computing graph automorphism partition is GI-complete~\citep{Booth1977:Problems}, we know that $\anrposs$ for simple $\committee 2$ preferences and $\committee 1$ decisions is GI-easy. The following theorem, which is a special case of Theorem~\ref{thm:GI-hardness-MM}, shows that the problem is in fact GI-complete.

\begin{customthm}{\ref{thm:GI-hardness-MM}$^-$}[{\bf\boldmath The $(\committee 2,\committee 1)$ case of Theorem~\ref{thm:GI-hardness-MM}}{}]
\label{thm:GIC-C2C1}
$\anrposs[(\committee 2,\committee 1)]$   is GI-complete.
\end{customthm}
\begin{sketch}
The GI-easiness of the problem follows after Lemma~\ref{lem:simple-M2} (i) and Proposition~\ref{prop:anr-possibility}. To prove the GI-hardness, we give a Turing reduction that uses polynomially many calls to $\anrposs$ to solve $\GI$. Let  $(G_1,G_2)$, which are graphs over $[m]$, denote a $\GI$ instance.  Let $G_1^*$ denote the graph that is obtained from $G_1$ by appending a cycle of length $m+1$ to vertex $1$. For every $i\in [m]$, let $G_2^i$ denote the graph that is obtained from $G_2$ by appending a cycle' of length $m+1$ to vertex $i$. For example, Let $G_1$ and $G_2$ denote the graphs in Example~\ref{ex:graph}. $G_1^*,G_2^1,G_2^2,G_2^3,G_2^4$ are shown in the following figure.
\begin{figure}[htp]
\includegraphics[width =\textwidth]{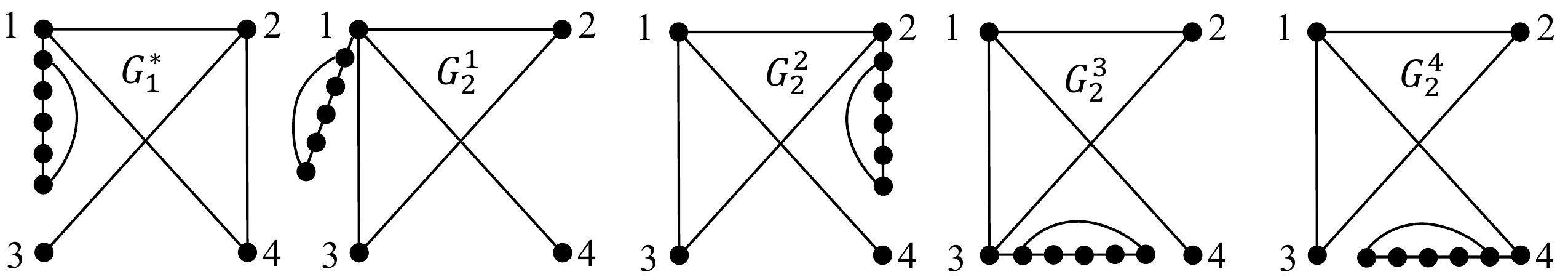}
\end{figure}

For every $i\in [m]$, let  $P_i$ denote any $\committee 2$-profile  whose histogram is $\vec h_{G_1^*+G_2^i}$ according to Definition~\ref{dfn:M2-hist-graph}, where $G_1^*+G_2^i$ is the disjoint union of  $G_1^*$ and $G_2^i$. We then solve $\anrposs[(\committee 2,\committee 1)]$ for $P_1,\ldots,P_m$, and if the answer is NO for any of them, then we output YES for the $\GI$ instance; otherwise  we output NO. The correctness of the proof can be found in Appendix~\ref{app:proof-thm:GIC-C2C1}.
\end{sketch}

\section{General $(\prefspace,\decspace)$ Settings: Challenges and Solutions}
\myparagraph{$\CL$-Easiness.}  A natural idea to extend the connection revealed (Lemma~\ref{lem:simple-M2})  to general (non-simple) $\committee \ell$-preferences is  to construct a multi-hyper-graph by converting a  committee preference $A$ to a hyper-edge that consists of vertices in $A$ with the same multiplicity. However, this idea does not work because it is unclear whether $\ap$ and $\rs$ ($\cl$) are preserved, and whether efficient canonical labeling algorithms exist for multi-hyper-graphs. 
We address this challenge by constructing an almost bipartite graph for the more general $\commonpref m$ preferences, where one side represents the vertices (denoted by $\ma$)  and the other side  (denoted by $\calN$) represents the $n'$ types of preferences in the given histogram $\vec h$  to represent the $n'$ types of preferences in $\vec h$. The connections between $\ma$ and each $R\in \mn$ (denoted by $\calV$) represent  $R$'s type, and $R$ is connected to a tail (denoted by $\calT$) whose length equals to the multiplicity of $R$ in $\vec h$. We also have two sets of auxiliary vertices $\calX$ and $\calY$, which are used to identify the two sides. The formal definition can be found in Def.~\ref{dfn:histogram->graph} in Appendix~\ref{sec:general-easy}.  An example of $\commonpref 4$-histogram $\vec h$ and $G_{\vec h}$ are shown in Figure~\ref{fig:ex-histogram-mapping}.

\begin{figure}[htp]
\centering 	
\includegraphics[width = .7\textwidth]{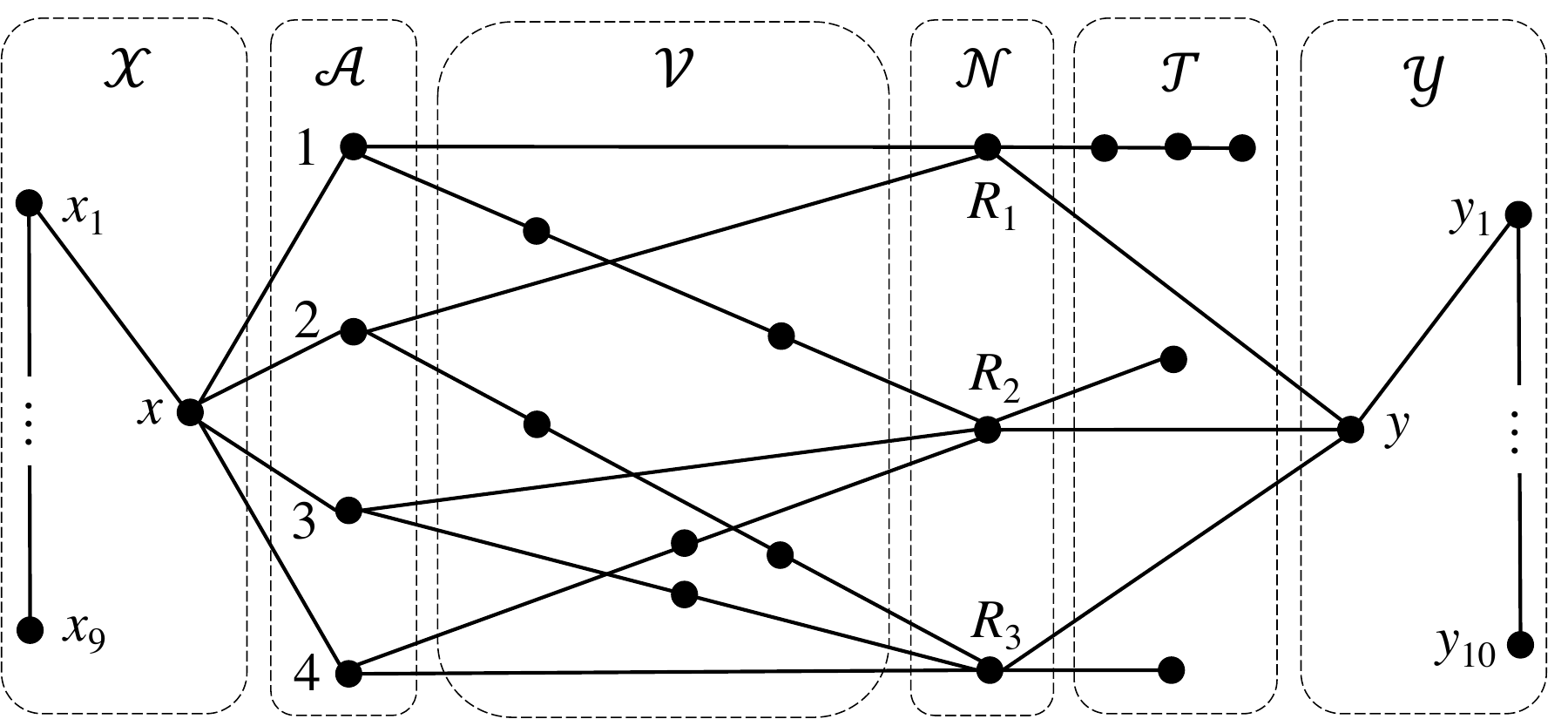}
\caption{\small $G_{\vec h}$ for $\vec h=(4\times  {R_1}, 2\times {R_2}, 2\times {R_3})$, where $R_1 = {\{1,2\}}, R_2 = {[ 3\succ 4\succ 1]} , R_3 =  {[4\succ 3\succ 2]}$, $m=4$, $n=8$. \label{fig:ex-histogram-mapping}}
\end{figure}
Then, we take advantage of the specific structures of the graph in Definition~\ref{dfn:histogram->graph} to establish a one- way connection from graph $\ap$ to histogram $\ap$, and a one-way connection from $\cl$ to $\rs$ in Lemma~\ref{lem:connection-general} in Appendix~\ref{sec:general-easy}.  These connections play a central role in designing efficient algorithms for $\merv$ and $\merev$ under general settings, in Alg.~\ref{alg:CLVR-general} and Alg.~\ref{alg:CLTB-general} in Appendix~\ref{sec:CLVR-general}, respectively.

\myparagraph{$\GI$/$\GA$-Hardness.} Recall that the $\GI$-hardness proof for Theorem~\ref{thm:GI-hardness-MM}$^-$ converts any $\GI$ instance to multiple $\anrposs[(\committee 2,\committee 1)]$ instances. However, converting an $\GI$ instance to multiple $\anrposs[(\prefspace,\decspace)]$ instances is more challenging than it appears due to the lack of two-way correspondence between $\commonpref m$-histograms and graphs. The challenge will be handled by case-by-case constructions in Appendix~\ref{sec:general-hard}.

\section{Future Work}
There are a few social choice settings under which the complexity of $\anrposs$ are open questions (marked by ``?'' in Table~\ref{tab:summary}). The exact complexity of computing a most equitable rule (without verification) and the exact complexity of computing a most equitable tie-breaking mechanism (without verification) are open questions. They are easier than $\CL$ and we conjecture that they are  $\GA$-hard or $\GI$-hard. We believe that extending the algorithms and complexity analysis to other social choice settings (such as unions of decision spaces in $\commondec m$), other desiderata  (such as Pareto optimality and monotonicity), and other problem settings (such as fair division) are interesting and important directions for future work.

\bibliographystyle{plainnat}
\bibliography{references}

\newpage
\tableofcontents
\appendix

\section{General Settings: $\CL$-Easiness}
\label{sec:general-easy}
{\bf The challenge.} There are two natural ideas to extend the connection revealed in the last section to general (non-simple) $\committee \ell$-preferences. The first idea is to construct a multi-hyper-graph by converting a  committee preference $A$ to a hyper-edge that consists of vertices in $A$ with the same multiplicity. The second idea is to construct a bipartite graph, where one side represents the vertices and the other side represents the votes, and there is an edge between $a$ on the left and $R$ on the right if and only if $a\in R$. However, none of them work directly, because it is unclear whether $\ap$ and $\rs$ ($\cl$) are preserved under such constructions, and whether efficient canonical labeling algorithms exist for multi-hyper-graphs. Moreover, it is unclear how the constructions can be done for $\listset\ell$-preferences and the even more general $\commonpref m$-preferences.

\subsection{The Connection: $\commonpref m$-Histograms $\Rightarrow$ Graphs}
\label{sec:connection-general}
Our construction is similar to the second idea discussed above. For any $\commonpref m$-histogram $\vec h$, we use one side (denoted by $\ma$) to represent  the alternatives and the other side (denoted by $\calN$) to represent the $n'$ types of preferences in $\vec h$. The connections between $\ma$ and each $R\in \mn$ (denoted by $\calV$) represent  $R$'s type, and $R$ is connected to a tail (denoted by $\calT$) whose length equals to the multiplicity of $R$ in $\vec h$. We also have two sets of auxiliary vertices $\calX$ and $\calY$, which are used to identify the two sides. Formally, the construction is defined as follows.
\begin{dfn}[{\bf\boldmath  ${\commonpref m}$-histograms$\Rightarrow$graphs}{}]
\label{dfn:histogram->graph}
Given an $\commonpref m$-histogram $\vec h = (h_1\times R_1,\ldots,h_{n'}\times R_{n'})$ with $\sum_{j=1}^{n'} h_{j}=n$, we define a graph $G_{\vec h}$ as follows:
\begin{align*}
\text{\bf Vertices:}&\begin{cases}
\text{$\ma = \{ 1,\ldots,  m\}$   represents the alternatives,}\\
\text{$\mn = \{R_1,\ldots, R_{n'}\}$   represents different types of votes in $\vec h$,}\\
\text{$\calV$: for each $j\le n'$, if $R_j\in \listset\ell$ for some $2\le \ell\le m$, then there are $(\ell-1)\ell/2$ vertices}\\
\hfill\text{ $\{v_j^{\ell',i'}:1\le i'\le \ell'\le \ell-1\}$}\\
\text{$\calT$: for each $j\le n'$, there are $h_j$ vertices $\{t_{j}^1,\ldots, t_j^{h_j}\}$,}\\
\text{Auxiliary vertices $\calX = \{x, x_1,\ldots,x_{n+1}\}$ and $\calY=\{y,y_1,\ldots,y_{n+2}\}$.}
\end{cases}
\\
\text{\bf Edges:}&\begin{cases}  
\text{Between $\ma$ and $\mn$: for every $j\le n'$.} \\
\text{\hspace{5mm}$\bullet$ if $R_j\in \committee \ell$, then for every $a_i\in R_j$,  there is an edge  $ i-v_j$;}\\
\text{\hspace{5mm}$\bullet$  if $R_j\in \listset \ell$, then for $i\in\ma$ that is ranked at the $\ell'$-th position in $R_j$,}\\
\text{\hspace{25mm}there is a path $i -v_j^{\ell',1}-\cdots-v_j^{\ell’, \ell'-1}-R_j$}\\
\text{On $\ma$'s side: $x$ is connected to all vertices in $\ma$, plus a tail $x-x_1-\cdots-x_{n+1}$}\\
\text{On $\mn$'s side: every $j\le n'$ is connected to a tail $R_j-t_j^1-\cdots-t_j^{h_j}$, and} \\ 
\text{\hspace{18mm}$y$ is connected to all vertices in $\mn$, plus a tail $y-y_1-\cdots-y_{n+2}$.}\\ 
\end{cases}
\end{align*}
\end{dfn}
In the  $\calV$ part, when $R_j\in\committee \ell$, $\{v_j^{\ell',i'}:1\le i'\le \ell'\}$ are used to represent the $(\ell'+1)$-th ranked alternative in $R_j$. If $R_j\in \committee \ell$, then $R_j$ is directly connected to all $a\in R_j$. In the $\calT$ part of  vertices, $\{t_{j}^1,\ldots, t_j^{h_j}\}$ represent the  multiplicity of $R_j$ in $\vec h$.  An example of $\commonpref 4$-histogram $\vec h$ and $G_{\vec h}$ are shown in Figure~\ref{fig:ex-histogram-mapping}.


Unlike Definition~\ref{dfn:M2-hist-graph} (for simple $\committee 2$-profiles), Definition~\ref{dfn:histogram->graph} is a one-way correspondence from $\commonpref m$-histogram $\vec h$ to graph $G_{\vec h}$ (which contains more vertices than the alternatives). Therefore, an automorphism partition  (respectively, a canonical labeling) for $G_{\vec h}$ cannot be directly used as an automorphism partition  (respectively, a representative selection function) for $\vec h$. 

We address this challenge by taking advantage of the specific structures of the graph in Definition~\ref{dfn:histogram->graph}. For the connection between automorphism partitions, notice that any automorphism $\mu$ for $G_{\vec h}$ must preserve the auxiliary nodes (i.e., $\calX\cup\calY$), and thus must map $\ma$ to itself. Therefore, the restriction of   graph automorphism partition $\ap(G_{\vec h})$ on $\ma$ is well-defined, which turns out to be an automorphism partition for $\vec h$ as proved later in Lemma~\ref{lem:connection-general} (i). 

The connection between the representative selection function and canonical labeling may seem straightforward in light of Definition~\ref{dfn:histogram->graph}, and a natural idea for obtaining a representative selection function for a given histogram $\vec h$ would be to first convert it to $G_{\vec h}$, apply  $\cl$, and then map it back to a canonical form of $\vec h$. However, because after applying $\cl$, $G_{\vec h}$ may not be a graph that corresponds to another histogram, it is unclear how the map-back part works. 

Our solution is based on the observation that, according to the construction in Definition~\ref{dfn:histogram->graph}, any graph $G$ that is isomorphic to $G_{\vec h}$ has certain structures that can be used to uniquely ``decode'' the alternatives, the preferences, and their multiplicity. This is formally presented in the following decoding algorithm, which relies on a lexicographic priority order $\rhd$ over $\commonpref m$.

\begin{algorithm}[htp]
\caption{{\bf (Decoding$_\cl$)} Given a $\commonpref m$-histogram $\vec h$, output a mapping $\sigma$ over $\ma$.\label{alg:decoding}}
\begin{algorithmic}[1] 
\STATE Compute $G_{\vec h}$ according to Definition~\ref{dfn:histogram->graph}, then compute $\mu = \cl(G_{\vec h})$ and $G=\mu(G_{\vec h})$.\label{step:compute-G}
\STATE Order vertices in $\mu(\ma)$ according to $\rhd$ and mark them $\{1,\ldots,m \}$.\label{step:identify-x}
\STATE For vertices in $\mu(\mn)$, use marks in step~\ref{step:identify-x} to identify the preferences, and then mark the vertices with $\{R_1,\ldots, R_{n'}\}$ according to $\rhd$.
\label{step:identify-v}
\STATE  Mark the vertices in $\mu(\calT)$ to identify multiplicities of the preferences.\label{step:identify-multiplicities}
\STATE Identify the histogram $\vec h^*$ using results of step~\ref{step:identify-v} and step~\ref{step:identify-multiplicities}.
\STATE {\bf return} the mapping from $\vec h$ to $\vec  h^*$ based on step~\ref{step:identify-x}.
\end{algorithmic}
\end{algorithm}
%
 
Using Algorithm~\ref{alg:decoding}, we establish the connection in the next lemma.

\begin{lem}[{\bf The connection}{}]  
\label{lem:connection-general}
For any $\commonpref m$-histogram $\vec h$ and any canonical labeling $\cl$, 
\begin{itemize}
\item [\bf (i) ] {\bf (Graph AP$\Rightarrow$Histogram AP)}: $\ap(\vec h) = \ap(G_{\vec h})|_{\ma}$;
\item [\bf (ii)] {\bf \boldmath ($\cl\Rightarrow$rep.~selection)}:  Algorithm~\ref{alg:decoding} computes a representative selection function for $\histset{m}{n}{\commonpref m}$.
\end{itemize}
\end{lem}

 \begin{proof}
To prove (i), let $\vec h =\vec h' = \hist(P)$ and consider the isomorphism from $G_{\vec h}$ to $G_{\vec h'}$. For any $a,b$ that belong to the same set in the partition $\ap(P)$, suppose isomorphism $\sigma$ (over $\ma$) maps $a$ to $b$. Then, by Proposition~\ref{prop:properties-GP} (iii), there exists an isomorphism $\mu$ that maps $G_{\vec h}$  to itself, i.e., $\mu$ is a automorphism, such that $\mu|_\ma  = \sigma$. This means that $a$ and $b$ are in the same set in the partition $\ap(G_{\hist(P)})|_{\ma}$.

For any $a,b$ that belong to the same set in the partition $\ap(G_{\hist(P)})|_{\ma}$, suppose isomorphism $\mu$ (from $G_{\vec h}$ to $G_{\vec h'}$) maps $a$ to $b$. Then, by Proposition~\ref{prop:properties-GP} (ii), $\mu|_{\ma}(a) = b$. This means that $a$ and $b$ are in the same set in the partition $\ap(P)$.

\myparagraph{(ii)} We first prove some properties of $G_{\vec h}$.

\begin{prop}[{\bf\boldmath Properties of $G_{\vec h}$}{}]
\label{prop:properties-GP}
For any pair of histograms $\vec h$ and $\vec h'$,

(1) $\vec h\cong\vec h' \Longleftrightarrow G_{\vec h}\cong G_{\vec h'}$;

(2) for any permutation $\mu$ such that $\mu(G_{\vec h}) = G_{\vec h'}$, we have $\mu(\ma) = \ma$ and $\mu(\mn)=\mn$. Let $\mu|_\ma$ be the restriction of $\mu$ on $\ma$. Then, $\mu|_\ma(\vec h) = \vec h'$. 

(3) for any isomorphism $\sigma$ such that $\sigma(\vec h) = \vec h'$, there exists an isomorphism $\mu$ such that $\mu(G_{\vec h}) = G_{\vec h'}$ and $\mu|_{\ma} = \sigma$.
\end{prop}
\begin{proof}
Part (1) $\Rightarrow$. For any $\sigma_\ma$ such that $\sigma_\ma(\vec h) = \vec h'$, we construct a mapping $\mu$ for $G_{\vec h}$ as follows.
\begin{itemize}
\item For any $i\in \ma$, $\mu(i) = \sigma(i)$.
\item For any $j\le n$, suppose $\sigma(v_{j})=v_{j'}$ for some $j'\le n'$. Then, define $\mu(v_j) = \sigma(v_{j})$, $\mu(t_j^s) = t_{j'}^s$, and $\mu(v_j^{\ell,\ell'})=v_{j'}^{\ell,\ell'}$.

\item For any $a\in\calD$, $\mu(a)=a$.
\end{itemize}

Part (1) $\Leftarrow$. It suffices to prove Part (2), which explicitly constructs an isomorphism $\mu|_{\ma}$ between $\vec h$ and $\vec h'$. Let us first make a few observations on $\mu$ in order. Because $G_{\vec h}$ and $G_{\vec h'}$ share the same set of alternatives $\ma$ and auxiliary vertices, we let $\ma\cup\mn\cup\calT\cup\calV\cup\calD$ denote the vertices in $G_{\vec h}$ and let $\ma\cup\mn'\cup\calT'\cup \calV'\cup\calD$ denote the vertices in $G_{\vec h'}$. First, notice that the tail attached to $y$ is the only tail with length $n+2$. Therefore, $\mu(y)=y$ and $\mu_{y_j}=y_j$. Then, among the remaining vertices, the tail attached to $x$ is the only tail of length $n+1$, which means that $\mu(x)=x$ and $\mu_{x_j}=x_j$. Notice that $\ma$ are the only remaining vertices that are directly attached to $x$. Therefore, $\mu(\ma) = \ma$. Consequently,  $\mu|_{\ma}$ is well-defined. 

Moreover, $\mu$ can be viewed as a mapping from $\vec h$ to $\vec h'$ in the following sense. First, similarly, $\mn' = \mu(\mn)=\mn$. Notice that $\calT$ consists of tails attached to $\mn$, they must be mapped to tails attached to $\mn'$ by $\mu$. This means that $\calT'=\mu(\calT)$. Then, for any $i\in\ma$ and $j\le n'$ with a path of length $s$, let $i' = \mu(i)$ and $v_{j'} = \mu(v_j)$, there is a path of length $s$ as well. This means that $\mu(\calV)=\calV'$, and $i$ being ranked at the $s$-th position in $R_j$ in $\vec h$ if and only if $i'$ being ranked at the $s$-th position in $R_{j'}$ in $\vec h'$, and $h_j = h_{j'}$ (because the lengths of tails attached to $v_j$ in $G_{\vec h}$ and to $v_{j'}$ in $G_{\vec h'}$ are the same). 

To prove Part (2), it remains to show that $\mu|_{\ma}(\vec h) =\vec h'$. Suppose for the sake of contradiction that this is not true. Then, there exists $j\le n'$ such that either 
\begin{itemize}
\item [] (i) $\mu|_{\ma}(R_j)$ does not appear in $\vec h$, i.e., $\mu|_{\ma}(R_j)\notin\{a\}$, or 
\item [] (ii) $\mu|_{\ma}(R_j)$ appears in $\vec h$, but the multiplicity does not match. That is, there exists $j'\le n'$ such that $R_{j'}=\mu|_{\ma}(R_j)$ but $h_{j'}\ne h_j$.
\end{itemize}
Either leads to a contradiction by considering $\mu(v_j)$ in $G_{\vec h'}$. More precisely, based on the observations on $\mu$ above, let $v_{j'}=\mu(v_j)$, which is a contradiction to (i). Also, the tails attached to $v_j\in G_{\vec h}$ must be mapped to the tail attached to $v_{j'}\in G_{\vec h'}$ by $\mu$, which is a contradiction to (ii). This proves Part (2) as well as (1) $\Leftarrow$.

Part (3) naturally follows after the definition in Definition~\ref{dfn:histogram->graph}:  $\mu$ does not change vertices in $\calD$. For any $a\in\ma\cup \mn$, let $\mu(a) = \sigma(a)$. Then, each vertex in $\calT\cup\calV$ (that is associated with vertex $v\in\mn$ and $i\in\ma$ are matched to corresponding vertices associated with $\mu(v)$ and $\mu(i)$. It is not hard to verify that $\mu$ is an isomorphism between $G_{\vec h}$ and $G_{\vec h'}$, and $\mu|_{\ma} = \sigma$.
\end{proof}

Part (1) of Proposition~\ref{prop:properties-GP} states that the conversion of histograms to graphs in Definition~\ref{dfn:histogram->graph} preserves the isomorphism relation. Part (2) states that any isomorphism between two graphs generated from histograms can be used to obtain an isomorphism between the two histograms. This is the key observation behind the algorithms in this section. Part (3) states that any isomorphism between the histograms can be obtained in a way described in Part (2).

Let $\sigma$ denote the output of Algorithm~\ref{alg:decoding}. It is not hard to verify that $G\cong G_{\sigma(\vec h)}$. It remains to show that for any pair of isomorphic histograms $\vec h\cong \vec h'$, Algorithm~\ref{alg:decoding} maps $\vec h$ and $\vec h'$ to the same histogram. Let $\sigma_{\vec h}$ and  $\sigma_{\vec h'}$ denote the output of Algorithm~\ref{alg:decoding} on $\vec h$ and $\vec h'$ respectively. It follows from Proposition~\ref{prop:properties-GP} that $G_{\vec h}\cong G_{\vec h'}$, which means that $\cl(G_{\vec h})= \cl(G_{\vec h'})$, which means that step~\ref{step:compute-G} of Algorithm~\ref{alg:decoding} computes the same graph $G$ on inputs $\vec h$ and $\vec h'$. Therefore, $\sigma_{\vec h}(\vec h) = \sigma_{\vec h'}(\vec h')$. This completes the proof of Lemma~\ref{lem:connection-general}.
\end{proof}

\subsection{Most Equitable Rule with Verifications and Tie-Breaking}
\label{sec:CLVR-general}

Based on  Lemma~\ref{lem:connection-general}, we propose $\CLR[\cl]$ for $\commonpref m$-preferences in  Algorithm~\ref{alg:CLVR-general}.

\begin{algorithm}[htp]
\caption{{\bf ($\CLR[\cl]$: $\commonpref m$-profiles)} Given  a $\commonpref m$-profile $P$, return $(d,c)\in\decspace \times \{0,1\}$.\label{alg:CLVR-general}}
\begin{algorithmic}[1]  
\STATE Compute $\vec h \triangleq \hist(P)$ and let $\rs^*$ denote the application of Algorithm~\ref{alg:decoding} to $\vec h$.
\STATE Compute $G_{\vec h}$ according to Definition~\ref{dfn:histogram->graph}. Compute $\ap(G_{\vec h})$ using $\cl$. Then, compute

\hfill$D^* \triangleq   \fw_\decspace(P) \text{ by applying Proposition~\ref{prop:verification} to $\ap(G_{\hist})$ }$\hfill
\STATE  {\bf if }$D^*=\emptyset$ {\bf then return }$(\arg\max_{d\in \decspace}^{\rhd}\rs^*(d), 0)$.
\STATE {\bf else return }$(\arg\max_{d\in D^*}^\rhd \rs^*(d), 1)$. 
\end{algorithmic}
\end{algorithm} 


\begin{thm}
\label{thm:CLVR-general}
For any graph canonical labeling  $\cl$, $\CLR[\cl]$ (Algorithm~\ref{alg:CLVR-general}) computes a $\merv$ under common profiles and common decisions with polynomially many calls to $\cl$.
\end{thm}

The canonical-labeling tie-breaking for $\commonpref m$-profiles is defined similarly. 

\begin{algorithm}[H]
\caption{{\bf ($\CLTB[\cl]$: $\commonpref m$-profiles)} Given a   $\commonpref m$-profile $P$ and  $D\subseteq\decspace$, return $(d,c)\in D \times \{0,1\}$.\label{alg:CLTB-general}}
\begin{algorithmic}[1]  
\STATE Compute $\vec h \triangleq \hist(P)$ and  a representative selection function $\rs^*$ by applying Algorithm~\ref{alg:decoding} to $\vec h$.
\STATE Compute $G_{\vec h}$ according to Definition~\ref{dfn:histogram->graph}. Compute $\ap(G_{\vec h})$ using $\cl$. Then, compute

\hfill$D^* \triangleq D\cap \fw_\decspace(P) \text{ by applying Proposition~\ref{prop:verification} to $\ap(G_{\vec h})$ }$\hfill 
\STATE  {\bf if }$D^*=\emptyset$ {\bf then return }$(\arg\max_{d\in \decspace}^{\rhd}\rs^*(d), 0)$.
\STATE {\bf else return }$(\arg\max_{d\in D^*}^\rhd \rs^*(d), 1)$. 
\end{algorithmic}
\end{algorithm}


\begin{thm}
\label{thm:CLTB-general}
For any graph canonical labeling $\cl$, $\CLTB[\cl]$ (Algorithm~\ref{alg:CLTB-general}) computes a $\merev$ for any anonymous and neutral irresolute rule under common profiles and common decisions with polynomially many calls to $\cl$.
\end{thm} 
As discussed earlier, using the quasipolynomial-time algorithm for graph canonical labeling by~\citet{Babai2019:Canonical} in Algorithm~\ref{alg:CLVR-general} and Algorithm~\ref{alg:CLTB-general}, we have the following corollary.
\begin{coro}
\label{coro:MERC-quasipoly}
Under common profiles and common decisions, $\merv$ and $\merev$ can be computed in quasipolynomial time.
\end{coro}

\subsection{Polynomial-Time Computable  Cases}
\label{sec:special-easy}
In this subsection, we present a few special cases of social choice settings under which $\merv$, $\merev$, and/or $\anrposs$ can be computed in polynomial time.  The first such case is bounded $m$. When $m$ is bounded above by a constant, the number of permutations over $\ma$ is a constant, which means that for any profile $P$, $\ap(P)$ can be computed in polynomial time according to Proposition~\ref{prop:verification}. Therefore, we have the following Theorem.

\begin{thm}
\label{thm:bounded-m}
For any bounded $m$, any $\decspace\in \commondec m$, and  any $\commonpref{m}$-profile,   $\merv$,  $\merev$, and $\anrposs$ can be computed in polynomial time. 
\end{thm}
The next theorem extends the polynomial-time algorithm for $\merev$ under $\listset m$ preferences~\citep{Xia2023:Most} to $\listset {m-C}$-profiles, where $C$ is a constant.  

\begin{thm}
\label{thm:const-unranked}
For any constant $C\in \mathbb Z_{\ge 0}$, and  any $\commonpref{m}$-profile, and any $\listset {m-C}$-profile,  $\merv$,  $\merev$, and $\anrposs$  can be computed in $C!\text{poly}(mn)$ time.
\end{thm}
\begin{proof}
We first review the algorithm for $(\listset{m},\decspace)$ in~\citep{Xia2023:Most}. For any $\listset m$-profile $P$, it adopts the lexicographic representative selection function that computes a permutation $\sigma$ that maps $\hist(P)$ to a histogram with the highest lexicographic priority. Such computation can be done in polynomial time because a linear order $R$ with the highest multiplicity in $\hist(P)$ must be mapped to the linear order with the highest priority, w.lo.g.~assumed to be $1\succ2\succ\cdots\succ m$. Consequently, we only need to examine $O(n)$ permutations to see which one maps $\hist(P)$ to the histogram with the highest priority.

Our algorithm for $\merv$ follows a similar idea to compute the lexicographic representative selection function for any $\listset {m-C}$-profile $P$. Notice that $(m-C)$-list with the highest multiplicity in $\hist(P)$ must be mapped to the $(m-C)$-list  with the highest priority, i.e., $1\succ \cdots\succ m-C$. Therefore, for each such  $(m-C)$-list we have $C!$ ways to specify the permutation for the remaining $C$ alternatives. Then, the number of permutations we need to consider is at most $C!O(n)$, which means that the algorithm runes in $C!O(mn)$ time.

$\merev$ and $\anrposs$ can be computed in similar ways. This proves Theorem~\ref{thm:const-unranked}.
\end{proof}

The following two theorems follow after straightforward applications of Proposition~\ref{prop:verification}.
\begin{thm}
\label{thm:iota=1}
For  any $m\ge 2$  and any $\decspace\in \commondec m$, under  $ (\committee{1},\decspace)$ (which is the same as $(\listset{1},\decspace)$) and $(\committee{m-1},\decspace)$,  $\merv$,  $\merev$, and $\anrposs$ can be computed in polynomial time. 
\end{thm}\begin{proof}
In light of Proposition~\ref{prop:verification}, it suffices to prove that in all settings described in the theorem and any profile, $\ap(P)$ can be computed in polynomial time. Let $P$ be a $\committee{1}$- or $\listset 1$-profile. It follows that two alternatives are in the same set of $\ap(P)$ if and only if their plurality score are the same, which means that $\ap(P)$ can be computed in polynomial time. The proof for  $\committee{m-1}$-profiles is similar and the only difference is that we check alternatives' veto scores instead. 
\end{proof}

 \begin{thm}
\label{thm:m-committee}
For any $m\ge 2$ and any $(\prefspace,\decspace)\in\CS{m}$, the answer to $\anrposs[(\prefspace,\committee m)]$ is always YES, and the answer to $\anrposs[(\committee m,\decspace)]$ is always NO except when $\decspace = \committee m$.
\end{thm}

\begin{proof}
The proof follows after a straightforward application of Proposition~\ref{prop:anr-possibility}. Specifically, when $\prefspace = \committee m$, for any profile $P$, every permutation is a stabilizer, which means that $\ap(P) = \{\ma\}$. The only case for $P$ to be ANR-possible is the case in which $\decspace = \committee m$.
\end{proof}

\section{General Settings: $\GI$/$\GA$-Hardness 
}
\label{sec:general-hard}
Recall that Theorem~\ref{thm:bounded-m} states that for constant $m$, $\merv$, $\merev$, and $\anrposs$ are in $P$. This motivates us to consider variable $m$ and series of social choice settings, one for each $m$. In practice, the preference spaces  across $m$'s are often of the same kind, i.e., all being $\ell$-lists  or all being $\ell$-committees, while $\ell$ is a function of $m$.  The same holds for the decision space. Therefore, in this section we investigate the complexity of $\anrposs$ for series of Common Settings defined as follows.

\begin{dfn}[{\bf $\anrposs$ for series of Common Settings}{}]
Given two functions $\iota:\mathbb N\ra \mathbb N$ and $\kappa:\mathbb N\ra \mathbb N$, in $\anrposs[(\committee \iota,\committee\kappa)]$ problem, the input consists of $m\ge 2$ and an $\committee{\iota(m)}$-profile $P$; and the output is YES if and only if $P$ is ANR-possible under $(\committee{\iota(m)},\committee{\kappa(m)})$. $\anrposs[(\committee \iota,\listset\kappa)]$, $\anrposs[(\listset \iota,\committee\kappa)]$, and $\anrposs[(\listset \iota,\listset\kappa)]$ are defined similarly. 
\end{dfn}

Recall from Proposition~\ref{prop:verification} that to compute $\anrposs$, it suffices to compute $\ap(P)$. In light of the connection between $\commonpref m$-histograms and graphs (Lemma~\ref{lem:connection-general}), computing $\ap(P)$ can be reduced to first computing $\ap(G_{\hist(P)})$, then computing its restriction to $\ma$, and finally verifying the conditions in Proposition~\ref{prop:verification}. This procedure was formally presented in Proposition~\ref{prop:anr-possibility}. Because computing $\ap$ for graphs is GI-complete, $\anrposs$ is $\GI$-easy as formally stated and proved as follows. 
\begin{thm}
\label{thm:GI-easiness}
For any $\scsetting\in\CS{m}$,  $\anrposs[\scsetting]\reducesto\GI$. 
\end{thm}

\begin{proof} 
The theorem is proved by the polynomial-time Turing reduction in Algorithm~\ref{alg:cerificate}.
\begin{algorithm}[htp]
\caption{{\bf ($\anrposs$)} Given an $\commonpref m$-profile $P$ and any $\decspace\in \commonpref m$, output YES if $P$ is an ANR-possible profile; otherwise output NO.\label{alg:cerificate}}
\begin{algorithmic}[1] 
\STATE  Convert $P$ to $G_{\hist(P)}$ by Definition~\ref{dfn:histogram->graph}.
\STATE Compute $\ap(G_{\hist(P)})$ by polynomially many calls to $\GI$. Let $\ap(G_{\hist(P)})|_\ma$ denote its restriction to $\ma$.
\FOR {every $d\in\decspace$}
\STATE if $d$ and  $\ap(G_{\hist(P)})|_\ma$ satisfy conditions in Proposition~\ref{prop:verification}, then {\bf return} YES.
\ENDFOR
\STATE {\bf return} NO. 
\end{algorithmic}
\end{algorithm}

\ 
\end{proof}

Therefore, in the rest of the section we focus on proving $\GI$-hardness results.


\begin{thm}[{\bf\boldmath $\GI$-C: $\anrposs[(\committee \iota,\committee \kappa)]$}{}]
\label{thm:GI-hardness-MM}
For any  $\iota$ and $\kappa$ such that for all $m\ge 3$, $\iota(m)\ge 2$   and $\kappa(m) \le m-1$,  $\anrposs[(\committee \iota,\committee \kappa)]$ is GI-complete. 
\end{thm}
\begin{proof}Without loss of generality, we can assume that $\kappa(m) = m-\text{poly}(m)$. This is because, otherwise, we consider the problem of $m-\kappa(m)$, and according to Proposition~\ref{prop:verification}, a $\commonpref m$-profile $P$ is ANR-possible under $\decspace =\committee k$ if and only if it is ANR-possible under $\decspace =\committee {m-k}$.   We first present the hardness of $\anrposs[(\committee 2, \committee \kappa)]$ by a reduction from $\GI$ to illustrate the idea, then extend the proof to other cases.

\myparagraph{$\GI\reducesto\anrposs[(\committee 2, \committee \kappa)]$.} Given a $\GI$ instance $(G_1,G_2)$ over vertices $[m]$, let $m^* = (6m+1)^2$, $k = \min(\kappa(m^*),m^*- \kappa(m^*))$. We construct $m-1$ instances of $\anrposs[(\committee 2, \committee \kappa)]$ to solve $\GI$ in two cases, depending on the value of $k$.  

{\bf\boldmath Case 1: $k\le 6m$.} For every $i\in [2,m]$, we construct a graph $G_i^*$ that contains the following vertices and edges, illustrated in Figure~\ref{fig:app:GI-hardness-committee-dec-small}.
\begin{itemize}
\item A subgraph $G_1'$ that is obtained from $G_1$ by adding a cycle of $m$ new vertices, and all such vertices are connected to vertex $1$ in $G_1'$.
\item $k$ identical subgraphs $G_2^{i,1}, \ldots, G_2^{i,k}$, each of which is obtained from $G_2$ by adding a cycle of $m$ new vertices, and all such vertices are connected to vertex $i$.
\item A cycle of the remaining $m^* - 2m(k+1)$  vertices (denoted by $\calX$).
\end{itemize}
\begin{figure}[htp]
\centering 
\includegraphics[width = .9\textwidth]{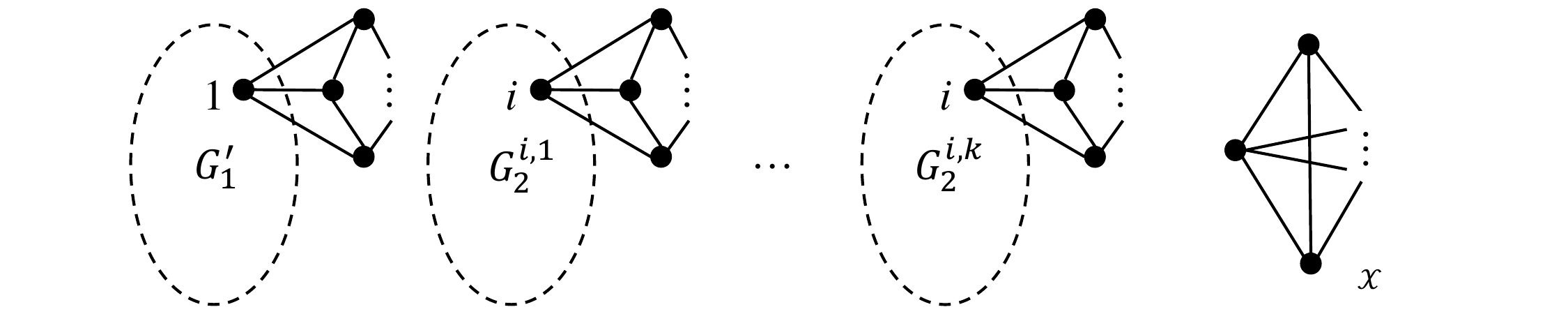}
\caption{$G_i^*$ for $k\le 6m$. \label{fig:app:GI-hardness-committee-dec-small}}
\end{figure}
Then, we convert $G_i^*$ to an $\committee 2$-profile, denoted by $P_i^*$, by Definition~\ref{dfn:M2-hist-graph}, and run $\anrposs[(\committee 2, \committee \kappa)]$ on $P_i^*$. If the answer to any of the $m-1$ instances is NO, then we output YES to the $\GI$ instance; otherwise, i.e., the answers to all  $m-1$ instances are YES, we output NO to the $\GI$ instance.

Clearly, the reduction runs in polynomial time. We use $j@G_1',j@G_2^{i,1},\ldots, j@G_2^{i,k}$ to denote the vertex labeled $j$ in the subgraph obtained from $G_1',G_2^{i,1},\ldots, G_2^{i,k}$,  respectively.  To verify the correctness of the reduction, we make the following observations on $P_i^*$. 

\begin{itemize}
\item If there is an isomorphism from $G_1$ to $G_2$ that maps $1$ to $i$, then $\{1@G_1', i@G_2^{i,1},\ldots, i@G_2^{i,k} \}$ is a set in $\ap(P_i^*)$, and the size of every set in $\ap(P_i^*)$, possibly except $\calX$, is divisible by $k+1$. Because $k\le 6m$, we have $|\calX| = m^* - 2m(k+1)\ge (6m+1)^2-2m(6m+1)>k$, which means that every set in $\ap(P^*)$ contains more than $k$ vertices, hence the $\anrposs$ instance $P_i^*$ is a NO instance according to Proposition~\ref{prop:anr-possibility}.
 
\item If there is no isomorphism from $G_1$ to $G_2$ that maps $1$ to $i$, then $\{i@G_2^{i,1},\ldots, i@G_2^{i,k} \}$ is a set (of size $k$) in $\ap(P_i^*)$. Therefore,  the $\anrposs$ instance $P_i^*$ is a YES instance according to Proposition~\ref{prop:anr-possibility}.
\end{itemize}

{\bf \boldmath Case 2: $k> 6m$.} For every $i\in [2,m]$, we construct a graph $G_i^*$ that contains the following vertices and edges, illustrated in Figure~\ref{fig:app:GI-hardness-committee-dec}.
\begin{itemize}
\item A subgraph $G_1'$ that is obtained from $G_1$ by adding a cycle of $m$ new vertices, and all such vertices are connected to vertex $1$ in $G_1'$.
\item Two identical subgraphs $G_2^{i,1}$ and $G_2^{i,2}$, each of which is obtained from $G_2$ by adding a cycle of $m$ new vertices, and all such vertices are connected to vertex $i$ in $G_2^{i,1}$ and $G_2^{i,2}$, respectively.
\item If $3\nmid m^*-6m-k+1$, then $G_i^*$ contains a cycle of $k-1$ new vertices (denoted by $\calX$) and a complete graph of $m^*-6m-k+1$ new vertices (denoted by $\calY$); otherwise $G_i^*$ contains a cycle of $k-2$ new vertices (denoted by $\calX$) and a complete graph of $m^*-6m-k+2$ new vertices (denoted by $\calY$).
\end{itemize}
\begin{figure}[htp]
\centering 
\includegraphics[width = .9\textwidth]{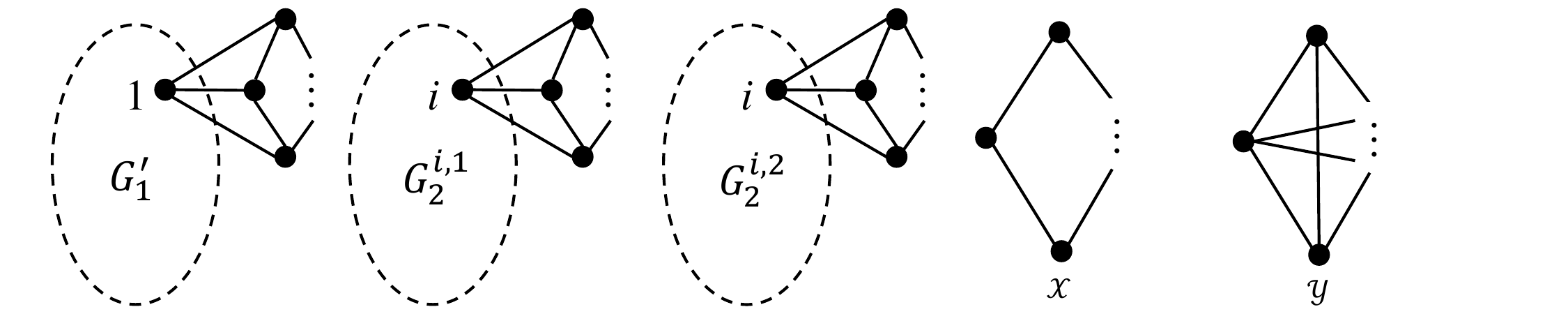}
\caption{$G_i^*$ for $k> 6m$. \label{fig:app:GI-hardness-committee-dec}}
\end{figure}
Then, we convert $G_i^*$ to an $\committee 2$ profile, denoted by $P_i^*$, by Definition~\ref{dfn:M2-hist-graph}, and run $\anrposs[(\committee 2, \committee \kappa)]$ on $P_i^*$. If the answer to any of the $m-1$ instances is NO, then we output YES to the $\GI$ instance; otherwise we output NO to the $\GI$ instance.

Clearly the reduction runs in polynomial time. To verify its correctness, we make the following observations about $P_i^*$. 

\begin{itemize}
\item If there is an isomorphism from $G_1$ to $G_2$ that maps $1$ to $i$, then $\calX$ and $\calY$ are two sets in $\ap(P^*)$, and   the size of every set in $\ap(P^*)$ is divisible by $3$.  Note that the construction guarantees that $3\nmid (k-|\calX|)$, $3\nmid (k-|\calY|)$, and we $k\le m^*/2$, which means that $|\calX|  + |\calY|>k$. Therefore, no subvector of $\apv(P^*)$ sums up to $k$, which means that the $\anrposs$ instance is a NO instance according to Proposition~\ref{prop:anr-possibility}.
 
\item If there is an isomorphism from $G_1$ to $G_2$ that maps $1$ to $i$, then $\{1@G_1'\}$ and $\{ i@G_2^{i,1}, i@G_2^{i,2} \}$ are two sets in $\ap(P^*)$, and $\calX$ and $\calY$ are two sets in $\ap(P^*)$. Therefore,  the subvector corresponds to $\calX$ and $\{1@G_1'\}$ or $\{ i@G_2^{i,1}, i@G_2^{i,2} \}$ sums up to $k$, which means that the $\anrposs$ instance is a YES instance according to Proposition~\ref{prop:anr-possibility}.
\end{itemize}

\myparagraph{$\GI\reducesto\anrposs[(\committee \iota, \committee \kappa)]$.} We first prove that, without loss of generality, we can assume that $\iota(m) = m-\text{poly}(m)$. This is because for any  $\committee \ell$ profile $P$,  let $\bar P$ denote the profile obtained from $P$ by taking the complements of all votes. For example, let $m=3$ and $P=(\{1,2\}, \{2,3\})$, then $\bar P = (\{3\}, \{1\})$. It is not hard to see that  $\ap(P) = \ap(\bar P)$. Therefore, to check wether an $\committee \ell$-profile $P$ is ANR-possible, it is equivalent to check the $\committee {m-\ell}$-profile $\bar P$ is ANR-possible.

 Let $m^* = m^{3/\epsilon}$ and let  $k = \min(\kappa(m^*),m^*- \kappa(m^*))$. We construct $m-1$ instances of $\anrposs[(\committee 2, \committee \kappa)]$ to simulate the two cases in the proof of the reduction for $\committee 2$ profiles for sufficiently large $m$.

{\bf\boldmath Case 1: $k\le 6m$.} We partition $[m^*]$ into the following sets.
$${1,\ldots, 2mk},\underbrace{2mk+1,\ldots, m^*}_{\calX}$$ 
Then, define $P^* = P_1\cup P_2$ as follows.
\begin{itemize}
\item Let $\hat G_i$ denote the subgraph of $G_i^*$ (defined in Case 1 of the $(\committee 2,\committee \kappa)$ part) that consists of $G_1'$, $G_2^{i,1},\ldots, G_2^{i,k}$.  For every $\{a,b\}\in \hat G_i$, there are $|\calX|$ votes in $P_1$
$$\{\forall i\le |\calX|, \sigma_{\calX}^i(\{a,b,2mk+1,\ldots, 2mk+\ell-2\})\}$$
\item $P_2$ consists of  
$$\{\forall i\le |\calX|, \sigma_{\calX}^i(\{2mk+1,\ldots, 2mk+\ell\})\}$$
\end{itemize}
$P^*$ is well defined because $2mk+\ell<m^*$. $P_1$ guarantees that the restriction of $\ap(P^*)$ to $[2mk]$ is the same as $\ap(\hat G_i)$, and $P_2$ guarantees that $\calX$ is a set in $\ap(P^*)$. The remaining proof is the same as that of Case 1 for the proof of the $\anrposs[(\committee 2, \committee \kappa)]$ case.

{\bf\boldmath Case 2: $k> 6m$.} We partition three sets, $[m^*]$ into $[6m]\cup\calX\cup\calY$, as follows.
\begin{align*}
\text{if $3\nmid m^*-6m-k+1$ , }&\underbrace{1,\ldots, 6m}_{[6m]},\underbrace{6m+1,\ldots, 6m+k-1}_{\calX}, \underbrace{6mk+k,\ldots, m^*}_{\calY}\\
\text{if $3\mid m^*-6m-k+1$ , }&\underbrace{1,\ldots, 6m}_{[6m]},\underbrace{6m+1,\ldots, 6m+k-2}_{\calX}, \underbrace{6m+k-1,\ldots, m^*}_{\calY}
\end{align*}  
Then, define $P^* = P_1\cup P_2$ as follows.
\begin{itemize}
\item Let $\tilde G_i$ denote the subgraph of $G_i^*$ (defined in Case 2 of the $(\committee 2,\committee \kappa)$ part) that consists of $G_1'$, $G_2^{i,1},G_2^{i,2}$.  For every edge $\{a,b\}\in \tilde G_i$, there are $|\calX|\cdot|\calY|$ votes in $P_1$
$$\{\sigma_{\calX}^i\circ \sigma_{\calY}^j(\{a,b,6m+1,\ldots, 6m+\ell-2\}): \forall i\le |\calX|, \forall j\le|\calY| \}$$
\item $P_2$ consists of $|P_1|+1$  copies of 
$$\{ \sigma_{\calX}^i\circ \sigma_{\calY}^j(\{6m+1,\ldots, 6m+\ell-1,m^*\}): \forall i\le |\calX|, \forall j\le|\calY|\}$$
\end{itemize}
$P_1$ guarantees that the restriction of $\ap(P^*)$ to $[2mk]$ is the same as $\ap(\tilde G_i)$, and $P_2$ guarantees that $\calX$ and $\calY$ are two elements in $\ap(P^*)$. The remaining proof is the same as that of Case 2 for the proof of the $\anrposs[(\committee 2, \committee \kappa)]$ case.
\end{proof}

Theorem~\ref{thm:const-unranked} motivates us to study list preferences in which a non-constant number of alternatives are unranked. We will primarily focus on such settings where the number of unranked alternatives is $\Omega(m^\epsilon)$ for some $\epsilon>0$.

\begin{thm}[{\bf\boldmath $\GI$-C: $\anrposs[(\committee \iota,\listset \kappa)]$}{}]
\label{thm:GI-hardness-ML}
For any  $\iota$ and $\kappa$ such that for all $m\ge 3$, $\iota(m)\ge 2$ and $\kappa(m) = m-\Omega(m^\epsilon)$ for some constant $\epsilon>0$,  $\anrposs[(\committee \iota,\listset \kappa)]$ is GI-complete.
\end{thm}
\begin{proof} As in the proof of Theorem~\ref{thm:GI-hardness-MM}, without loss of generality, we can assume that $\iota(m) = m-\text{poly}(m)$. 
We prove the hardness by presenting a Turing reduction from  $\anrposs[(\committee 2,\committee 1)]$, which is GI-complete (Theorem~\ref{thm:GIC-C2C1}). For any $\anrposs[(\committee 2,\listset 1)]$ instance $P$ over alterantives $[m]$, we choose the smallest $m^*$ such that $m^* - \iota(m)> m$ and $m^* - \kappa(m)> m$. It follows that $m^* = O(m^{1/\epsilon})$, which is polynomial in $m$. Let $\ell \triangleq \iota(m)$ and $k \triangleq \kappa(m)$.  For convenience, we partition $[m^*]$ into the following three sets $\ma\cup\calX\cup\calY$:
$$\underbrace{1,\ldots, m}_{\ma},\underbrace{m+1,\ldots, m+k}_{\calX},\underbrace{m+k+1, \ldots, m^*}_{\calY}$$ 
Let $\sigma_{\calX}$ and $\sigma_{\calY}$ denote arbitrary cyclic permutations over $\calX$ and $\calY$, respectively. We construct a profile $P^*=P_1\cup P_2
\cup P_3$, defined as follows.
\begin{itemize}
\item For every $\{a,b\}\in P$, there are $|\calY|$ votes in $P_1$
$$\{\forall i\le |\calY|, \sigma_{\calY}^i(\{a,b,m+1,\ldots, m+\ell-2\})\}$$
\item $P_2$ consists of votes that only approve subsets of $\calX\cup\calY$ and guarantees that alternatives in $\calX$ receive  different number of votes in $P_1\cup P_2$, while alternatives in $\calY$ receive the same number of votes in $P_1\cup P_2$. This is based on the construction guaranteed by the following proposition.
\begin{prop}
\label{prop:existence-different-votes}
For any $m'\in\mathbb N$ and $\ell'<m'$, there exists a profile $P'$ of no more than $(m')^2$ such that all alternatives receive different votes.
\end{prop}
\begin{proof}
The prove is done by induction on $\ell'$. Clearly, when $\ell'=1$ the proposition holds for all $m'$. Suppose the proposition holds for $\ell'=t$. When $\ell'=t+1$, we let all votes contain alternative $1$, and start with $\lceil \frac{m'-1}{m'-t-1}\rceil$ votes such that each alternative, except $1$, is not approved in at least one vote. We then rename the remaining alternatives in the decreasing order of votes, and apply the induction hypothesis on alternatives $\{2,\ldots,m'\}$ for $\ell'=t$. Let $P''$ denote the votes. We then add $1$ to all votes in $P''$ and merge the result to $P'$. It is not hard to verify that all alternatives receive different numbers of votes in $P'$, and $|P'|=\lceil \frac{m'-1}{m'-t-1}\rceil+|P''|<m'-1+(m'-1)^2<(m')^2$. This completes the proof.
\end{proof}
Let $P'$ denote the profile guaranteed by Proposition~\ref{prop:existence-different-votes} on $\{m+1,\ldots,m^*\}$ and $\ell$, and w.l.o.g.~let alternative $m+1$ receives the most votes, $m+2$ receives the second most votes, etc. Then, we let
$$P_2\triangleq \bigcup\nolimits_{i\le |\calY|} \sigma_{\calY}^i(P')$$
It follows that alternatives in $\calX$ receive different numbers of votes, alternatives in $\calY$ receive the same number of votes, and the former are larger than the latter.
\item $P_3$ is used to guarantee that alternatives in $\calY$ receive more votes than any alternative in $[m]$. It consists of the following $n\cdot|\calX|\cdot|\calY|$ votes 
$$n\times \left \{\forall i\le   |\calY|, \forall j\le   |\calX|,   \sigma_{\calY}^i\circ \sigma_{\calX}^j(\{m+1,\ldots,m+\ell-1,m^*\})\right
\}$$
\end{itemize}
The construction guarantees that $\ap(P^*)$ consists of every alternative in $\calX$ is a singleton, $\calY$, and sets in $\ap(P)$.
\begin{prop}
\label{prop:properties-P-star}
$\ap(P^*)= \calX\cup \{\calY\}\cup \ap(P)$. 
\end{prop}
\begin{proof}
We have the following observations about $P^*$. 
\begin{itemize}
\item[] (i) Alternatives in $\calX$ receive different numbers of votes, which are higher than that of any other alternative.
\item[] (ii) Alternatives in $\calY$ are ``similar'', because $\sigma^i$ (for all $i\le |\calY|$) are applied in  $P_1$, $P_2$, and $P_3$. 
\item[] (iii) Only $P_1$ involves $\ma$, and $P_1$ is ``similar'' to $P$ after removing all alternatives that are not in $\ma$. 
\end{itemize}
(i) implies that alternatives in $\calX$ are singletons in $\ap(P^*)$. (ii) implies that alternatives in $\calY$ are in the same set---more formally, this holds because $\sigma_{\calY}(\hist(P^*)) = \hist(P^*)$. Because every alternative in $\calY$ receives more votes than any alternative in $\ma$, no stabilizer $\sigma^*$ of $\hist(P^*)$ maps an alternative in $\ma$ to an alternative in $\calY$. This means that $\calY\in \ap(P^*)$. Following a similar reasoning, every stabilizer $\sigma^*$ of $\ap(P^*)$ maps $\ma$ to $\ma$, and its restriction to $\ma$ is a stabilizer of $\hist(P)$. On the other hand, for any stabilizer $\sigma$ of $\hist(P)$, we extend it to a permutation $\sigma^*$ over $[m^*]$ by letting (1) for every $a\in\ma$, $\sigma^*(a)=\sigma(a)$; (2) for every $x\in\calX$, $\sigma^*(x) = x$; and (3) for every $y\in\calY$, $\sigma^*(y) = \sigma_{\calY}(y)$. It follows that $\sigma^*$ is a stabilizer of $\hist(P^*)$. This completes the proof.
\end{proof} 
According to Proposition~\ref{prop:properties-P-star}, $\ap(P^*)$ has at least $k$ singletons ($\calX$ plus singletons in $\ma$) if and only if $\ap(P)$ has at least one singleton. Therefore, according to Proposition~\ref{prop:anr-possibility}, $P^*$ is a YES instance of $\anrposs[(\committee \iota,\listset \kappa)]$ if and only $P$ is a YES instance of of $\anrposs[(\committee 2,\listset 1)]$.  Also, notice that $|P^*| =\text{poly}(|P|)$. This proves $\anrposs[(\committee 2,\listset 1)]\reducesto\anrposs[(\committee \iota,\listset \kappa)]$.
\end{proof}


\begin{thm}[{\bf\boldmath $\GI$-C: $\anrposs[(\listset \iota,\committee \kappa)]$}{}]
\label{thm:GI-hardness-LM}
For any  $\iota$ and $\kappa$ such that for all $m\ge 3$, $2\le \iota(m)\le m-\Omega(m^\epsilon)$ for some constant $\epsilon>0$,   and $\kappa(m) \le m-1$,  $\anrposs[(\listset \iota,\committee \kappa)]$ is GI-complete. 
\end{thm}
\begin{proof}
The construction is similar to the proof of the Theorem~\ref{thm:GI-hardness-MM}. The only difference is that when defining $P^*$, we rank alternatives in alphabetically. 
\end{proof}

\begin{thm}[{\bf\boldmath $\GI$-C: $\anrposs[(\listset \iota,\listset \kappa)]$}{}]
\label{thm:GI-hardness-LL}
For any  $\iota$ and $\kappa$ such that for all $m\ge 3$, $2\le \iota(m)\le m-\Omega(m^\epsilon)$   and $\kappa(m) = m-\Omega(m^\epsilon)$ for some constant $\epsilon>0$,  $\anrposs[(\listset \iota,\listset \kappa)]$ is GI-complete. 
\end{thm}
\begin{proof}
The construction is similar to the proof of the Theorem~\ref{thm:GI-hardness-ML}. The only difference is that when defining $P^*$, we rank alternatives in the $k$-committees alphabetically before applying $\sigma_{\calX}$ and/or $\sigma_{\calY}$. 
\end{proof}

\begin{thm}[{\bf\boldmath $\GA$-H and $\GA$-C}{}]
\label{thm:GA-H} 
For any function $\iota$ and $\kappa$ such that there exists $0<\epsilon$ such that for all $m\ge 3$, $2\le \iota(m)\le m-\Omega(m^\epsilon)$ and $\kappa(m) = \Omega(m^\epsilon)$ for some constant $\epsilon>0$, 
$$\GA\reducesto\anrposs[(\committee \iota, \listset \kappa)] \text{, and } \GA\reducesto\anrposs[(\listset \iota, \listset \kappa)] $$
Moreover,  $\anrposs[(\commonpref m, \listset {m})]\reducesto \GA \text{, and }\anrposs[(\commonpref m, \listset {m-1})]\reducesto \GA$.
\end{thm}
\begin{proof} We present the $\GA$-hardness for $(\committee 2, \listset \kappa)$ and comment on how to extend the proof to other cases. 

\myparagraph{Hardness for $(\committee 2, \listset \kappa)$.} We provide a polynomial-time Turing reduction that solves a $\GA$ instance by converting it to an $\anrposs[(\committee 2, \listset \kappa)]$ instance. In fact, the reduction is a polynomial-time many-one reduction. For any $\GA$ instance $G$ with $m$ vertices, let $m^*\in\mathbb N$ be the smallest number such that $\kappa(m^*)\ge 2m+5$. We construct a graph $G^*$ from $G$ by adding the following new vertices and new edges, as illustrated in Figure~\ref{fig:ex-GA}.

\begin{figure}[htp]
\centering 
\includegraphics[width = .9\textwidth]{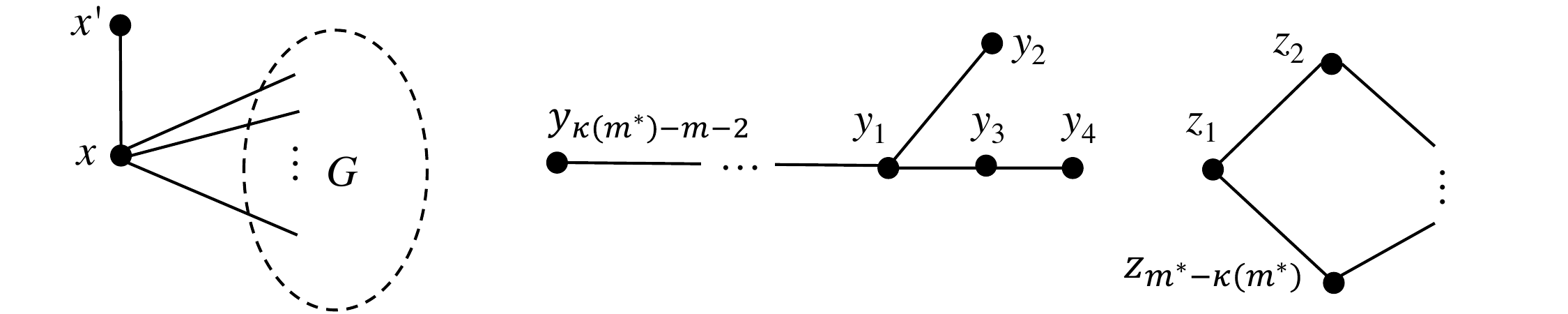}
\caption{Construction of $G^*$ from $G$. \label{fig:ex-GA}}
\end{figure}

\begin{itemize}
\item New vertices: $x,x'$. New edges: $x$ is connected to $x'$ and all vertices in $G$.
\item New vertices: $\calY = \{y_1,\ldots,y_{\kappa(m^*)-m-2}\}$. New edges: $\{y_1,y_2\},\{y_1,y_3\}, \{y_3,y_4\},\{y_1,y_5\},\{y_5,y_6\}$, $\ldots, \{y_{\kappa(m^*)-m-3},y_{\kappa(m^*)-m-2}\}$
\item New vertices: $\calZ = \{z_1,\ldots,z_{m^*-\kappa(m^*)}\}$. New edges: a cycle among these vertices.
\end{itemize}
Then, we construct a profile $P^*$ such that $\hist (P^*) = \vec h_{G^*}$, and run $\anrposs[(\committee 2, \listset \kappa)]$ on $P^*$. If the answer is YES, then we output YES to the $\GA$ instance; otherwise, we output NO.

Because $m^* = O(m^{1/\epsilon})$, the reduction takes polynomial time. To see its correctness, we note that for every automorphism $\sigma$ of $G^*$, we have $\sigma(x) = x$; $\sigma(x') = x'$; for every $y\in \calY$, $\sigma(y)\in\calY$; and for every $z\in\calZ$, $\sigma(z)\in\calZ$. Moreover, every $y\in \calY$ is a singleton in $\ap(G^*)$ and all vertices in $\calZ$ constitutes a set in $\ap(G^*)$. It follows that the number of singletons in $\ap(G^*)$ is the same as the number of singletons in $\ap(G)$ plus $\kappa(m^*)-m$ (i.e., $\calZ$, $x$, and $x'$). Consequently, according to Proposition~\ref{prop:anr-possibility}, the answer to $\anrposs$ is YES if and only if the number of singletons in $\ap(P^*)=\ap(G^*)$ is no more than $\kappa(m^*)-1$, which is equivalent to the number of singletons in $\ap(G)$ being no more than $m-1$, i.e., $G$ has a non-trivial automorphism.

\myparagraph{Hardness for $(\committee \iota, \listset \kappa)$ and $(\listset \iota, \listset \kappa)$.} The proof is done by showing a polynomial-time Turing reduction from the $(\committee 2, \listset \kappa)$ case  in the same way as in the  proof of Theorem~\ref{thm:GI-hardness-ML} and the proof of Theorem~\ref{thm:GI-hardness-LM}. The construction guarantees that $\ap(G)$ in the $(\committee 2, \listset \kappa)$ case is preserved, alternatives in $\calY$ are singletons, and alternatives in $\calZ$ are non-singletons.

\myparagraph{$\GA$-easiness.} For any instance $P$, we convert it to $G_{\hist(P)}$ using Definition~\ref{dfn:histogram->graph}. Then, we ask $\GA$ on $G$---if the answer is YES, then the answer to $P$ is NO; otherwise, the answer to $P$ is YES. 
\end{proof}

\section{Proof of the Correctness of Theorem~\ref{thm:GI-hardness-MM}$^-$}
\label{app:proof-thm:GIC-C2C1}
Clearly, the reduction is a polynomial-time Turing reduction. To see its correctness, suppose the $\GI$ instance is a YES instance. Then, there exists a permutation $\mu$ over $[m]$ such that $\mu(G_1) = G_2$. Then, $P_{\mu(1)}$ is a NO instance of $\anrposs[(\committee 2,\committee 1)]$ because $P_{\mu(1)}$ does not have a fixed-point decision (which is equivalent to $\ap(P_{\mu(1)})$  not  containing a singleton according to Proposition~\ref{prop:verification}): notice that  under the isomorphism $\mu$, where for each $i\in [m]$, $i$ in $G_1^*$ is mapped to $\mu(i)$ in the $G_2^{\mu(1)}$ part of the graph; and   the cycle of length $m+1$ in $G_1^*$ is mapped the cycle of length $m+1$ in the $G_2^{\mu(1)}$ part of the graph.

Conversely, suppose there exists $i^*\in [m]$ such that $P_{i^*}$ is a NO instance of $\anrposs[(\committee 2,\committee 1)]$. This means that alternative $1$ in $G_1^*$ is not a fixed-point decision, or equivalently, not a singleton in $\ap(P_{i^*})$ according to Proposition~\ref{prop:verification}. Due to the cycle attached to it, the only other alternative it can be mapped to is alternative $i^*$ in the $G_2^{i^*}$ part of the graph. Therefore, there exists a permutation $\sigma$ over the alternatives (all vertices in $G_1^*$ and all vertices in $G_2^{i^*}$) such that $i$ in $G_1^*$ is mapped to $i^*$ in $G_2^{i^*}$.   It follows that each vertex in $G_1^*$ must be mapped to a vertex in $G_2^{i^*}$ while maintaining the connectivity, which means that  $\sigma$ can be viewed as an isomorphism between $G_1$ and $G_2$, and hence the $\GI$ instance is a YES instance.

\end{document}